\documentclass{article}
\usepackage{style}

\usepackage[square,numbers, sort]{natbib}

\title{How to build a cross-impact model from first principles: \\
Theoretical requirements and empirical results \thanks{We warmly thank J.P. Bouchaud, Z. Eisler, B. T\'oth, M. Rosenbaum and A. Fosset for fruitful discussions. This research was conducted within the \textit{Econophysics \& Complex Systems} Research Chair under the aegis of the Fondation du Risque, a joint initiative by the \textit{Fondation de l'\'Ecole polytechnique, l'\'Ecole polytechnique} and Capital Fund Management. M. Tomas also acknowledges the support of the chairs \textit{Analytics and models for financial regulation}, \textit{Deep finance and statistics} and \textit{Machine learning and systematic methods.}}}

\author[1,2,3]{\normalsize Mehdi Tomas}
\author[3,4]{Iacopo Mastromatteo}
\author[1,3,4]{Michael Benzaquen}
 
\affil[1]{\normalsize LadHyX UMR CNRS 7646, Ecole Polytechnique, 91128 Palaiseau Cedex, France}
\affil[2]{\normalsize CMAP UMR CNRS 7641, Ecole Polytechnique, 91128 Palaiseau Cedex, France}
\affil[3]{\normalsize Chair of Econophysics \& Complex Systems, Ecole Polytechnique, 91128 Palaiseau Cedex, France}
\affil[4]{\normalsize Capital Fund Management, 23-25, Rue de l’Universit\'e 75007 Paris, France}

\begin{document}
\renewcommand{\arraystretch}{0.8}
\maketitle

\begin{abstract}
Trading a financial instrument pushes its price and those of other assets, a phenomenon known as cross-impact. To be of use, cross-impact models must fit data and be well-behaved so they can be applied in applications such as optimal trading. To address these issues, we introduce a set of desirable properties which constrain cross-impact models. We classify cross-impact models according to which properties they satisfy and stress them on three different asset classes to evaluate goodness-of-fit. We find that two models are robust across markets, but only one satisfies all desirable properties and is appropriate for applications.
\end{abstract}



\section*{Introduction}

Trading pressure moves prices, a now well-established phenomenon known as \textit{market impact} \cite{Bouchaud2018TradesPrices,Almgren2005DirectImpact,Torre1997BARRAHandbook}. In fact, market impact has been measured in many independent studies and is robust across assets, time periods and markets. A more subtle effect is that trading pressure from one asset can move the price of another. This effect has been dubbed \textit{cross-impact}. Cross-impact transmits information across assets and amplifies market shocks. Many papers incorporate cross-impact in applications but assume that the parameters of the model are known \cite{ekren2018optimal,ekren2019portfolio,garleanu2016dynamic,tsoukalas2019dynamic}. To be applied in practice, they require a calibration technique to estimate cross-impact from data. 
\\ \\
The importance of cross-impact has sparked recent interest in devising calibration methodologies from empirical data \cite{hasbrouck2001common,pasquariello2015strategic,Benzaquen2017DissectingAnalysis,wang2016cross,schneider2017,tomas2021cross,rosenbaum2021characterisation}. However, empirical studies focus on goodness-of-fit to calibrate cross-impact models or make simplifying assumptions which may work well on certain asset classes but break on others. Therefore, from the literature, we cannot determine whether there exists a universally robust and statistically accurate cross impact model, nor more specifically one which is suitable in other applications, such as optimal trading.
\\ \\
This paper seeks to bridge this gap by introducing desirable properties of cross-impact models, classifying models according to which properties they satisfy and stress-testing these models on different markets to assess which provide good empirical fit. In particular, we find that there is a single cross-impact model robust across asset classes and well-suited for applications.
\\ \\
We briefly comment on the links between our approach and the literature.
\\ \\
The paper most related to ours is \cite{rosenbaum2021characterisation}. There, the authors characterise suitable cross-impact models within a market where trades are modeled with Hawkes processes. The resulting cross-impact models are well-behaved and can be calibrated, which the authors illustrate on E-Mini S\&P 500 futures.
\\ \\
This paper stress-tests many different cross-impact models and thus contributes to the literature on the calibration and goodness-of-fit testing of cross-impact models \cite{Benzaquen2017DissectingAnalysis,schneider2017,rosenbaum2021characterisation}. However, papers on the literature focus on one (or few) models at a time on a particular asset class. Thus, this paper shines light on which models perform best on a variety of markets.
\\ \\
A set of theoretical studies have attempted to characterise suitable cross-impact models from certain properties. Unfortunately, a pure no-arbitrage framework as in \cite{Alfonsi2016MultivariateFunctions,schneider2017} is not sufficiently restrictive to prescribe a calibration methodology. We use some of their results to classify static cross-impact models which allow for arbitrage.
\\ \\ 
Other studies obtain cross-impact models via interactions of different agents. While the mean-field framework for optimal execution of \cite{Lehalle2019ACorrelations} provides one explanation of the many possible phenomena underlying cross-impact, it does not provide a recipe one may use on empirical data. In the optimal market making literature, \cite{bergault2018closed} finds that the liquidation costs of a market maker faces when he holds an inventory $q$ is of the form $- q^\top \Lambda q$, where $\Lambda$ can be estimated in practice. In fact, the same $\Lambda$ can be derived from the multivariate Kyle framework \cite{delMolino2018TheDifferent,Caballe1994ImperfectNeutrality}. We show that this $\Lambda$ plays a special role, as it is the only cross-impact model which satisfies all desirable properties. This partially explains why this cross-impact model appears in a variety of seemingly unrelated settings.
\\ \\
The paper is organized as follows. \cref{sec:setup_assumptions} introduces the setup of the paper. \cref{sec:axioms} lists axioms and \cref{sec:candidate_cross_impact_models} cross-impact models of interest, along with the axioms they satisfy. \cref{sec:illustrative_examples_empirical} stress-tests cross-impact models on a variety of markets. We conclude by stressing the main contributions of the paper and discussing open questions and directions for future work.

\section*{Notation}
\label{sec:notation}

The set of $n$ by $n$ real-valued square matrices is denoted by $\mat{n}$, the set of orthogonal matrices by $\orth{n}$, the set of real symmetric positive semi-definite matrices by $\spd{n}$ and the set of real symmetric positive definite matrices by $\pd{n}$. Given a matrix $A$ in $\mat{n}$, $A^\top$ denotes its transpose. Given $A$ in $\spd{n}$, we write $A^{1/2}$ for a matrix such that $A^{1/2} (A^{1/2})^\top = A$ and $\sqrt{A}$ for the matrix square root, the unique positive semi-definite symmetric matrix such that $(\sqrt{A})^2 = A$. We write $\ker{(M})$ for the null space of a matrix $M \in \mat{n}$, $\Pi_V$ for the projector on a linear subspace of $V \in \mathbb{R}^n$ and $\bar \Pi_V = \id - \Pi_V$ for the orthogonal projector. Finally, given a vector $\boldm{v} \in \reals{n}$, we write $\boldm{v} = (v_1, \dots, v_n)$ and $\diag{\boldm{v}}$ for the diagonal matrix with diagonal components the components of $\boldm{v}$.


\section{Cross-impact models as functions of market data}
\label{sec:setup_assumptions}

To relate trades to prices, we observe the mid-prices and trades of $\nassets$ different assets, both being binned on a regular time interval of length $\Delta t$. We denote by $\bp^i_t$ the opening price of Asset $i$ in the time window $[t, t+\Delta t]$ and by $\bp_t = (\bp^1_t, \cdots, \bp^\nassets_t)$ the vector of asset prices. We write $\bq^i_t$ the net market order flow traded in the same period, i.e. the signed sum of the volumes of all trades in that window, counting buy trades volume positively and sell trades volume negatively. Similarly, $\bq_t = (\bq^1_t, \cdots, \bq^\nassets_t)$ is the vector of net traded order flow.
\\ \\
On a given trading day, our goal is to relate the time series of net order flows $\{ \bq_0, \bq_{\Delta t}, \cdots, \bq_T\}$ with the time series of prices $\{\bp_0, \bp_{\Delta t}, \cdots, \bp_T \}$. For a given asset, it is classically admitted that price changes and net order flow are linearly related, although the linear relationship breaks down for large order flow values \cite{patzelt2018universal}. Inspired by this, we also assume this is true across assets. Furthermore, to emphasize the cross-sectional features of the problem, we discard the influence of past order flows. This leads to the following assumption.
\begin{assumption}
\label{ass:linear}
Price changes $\bdp_t := p_{t+\Delta t} - p_t$ and order flow imbalances $\bq_t$ are linearly related, i.e.
\begin{equation}
  \label{eq:lin_impact_intro}
\bdp_t = \Lambda_t \bq_t + \bmeta_t \, ,
\end{equation}
where the $\nassets \times \nassets$ matrix $\Lambda_t$ is called the cross-impact matrix and $\bmeta_t = (\bmeta^1_t, \dots, \bmeta^\nassets_t)$ is a vector of zero-mean random variables representing exogenous noise.
\end{assumption}

In \cref{eq:lin_impact_intro}, price changes and order flow imbalances are known and we have to choose the cross-impact matrix $\Lambda_t$. Our choice of cross-impact matrix influences the quality of fit of the model via the size of the residuals of \cref{eq:lin_impact_intro}. The focus of this paper is to find the right choice for the cross-impact matrix $\Lambda_t$.
\\ \\
While many choices for the cross-impact matrix are possible, we focus on those which depend on the parameters of the true data generating process. In the particular context where all random variables are Gaussian, then second-order statistics are sufficient statistics. This motivates the next assumption.

\begin{assumption}
\label{ass:sufficient_stats}
In the rest of the paper, we assume that price dynamics are given by
\begin{equation}
  \label{eq:lin_impact}
\bdp_t = \Lambda(\Sigma_t, \Omega_t, R_t) \bq_t + \bmeta_t \, ,
\end{equation}
where $\Lambda \colon \spd{n} \times \pd{n} \times \mat{n} \to \mat{n}$ is called a \emph{cross-impact model},  $\Sigma_t = \cov(\bdp_t)$ is the price change covariance matrix, $\Omega_t := \cov(q_t)$ is the order flow covariance matrix, $R_t = \E \left[ (\bdp_t - \E[\bdp_t])(q_t - \E[q_t])^\top \right]$ is the response matrix and $\bmeta_t = (\bmeta^1_t, \dots, \bmeta^\nassets_t)$ is a vector of zero-mean random variables representing exogenous noise.
\end{assumption}

The cross-impact model $\Lambda$ of \cref{eq:lin_impact} is a function of the three key second-order statistics which describe our market. Indeed, $\Sigma_t$ quantifies the  co-variation of returns, $\Omega_t$ captures co-trading of different assets, and $R_t$ reflects the average change of asset prices with traded order flow. By examining cross-impact models, we abstract the particular value of a cross-impact model for a given set of second-order statistics. Instead, this allows us to examine the cross-impact model across markets.
\\ \\
The main purpose of this paper is to find a suitable cross-impact impact model $\Lambda$ in \cref{eq:lin_impact} given a set of observations of market data and corresponding statistics $\Sigma_t, \Omega_t, R_t$. The next section discusses how to choose a proper cross-impact model $\Lambda$. First, we comment on the validity of our assumptions.

\begin{remark}[Validity of the static approximation]
Overall, \cref{ass:linear} is relevant to describe price impact shortly induced after trading for small portfolios. To assess its validity, we compare our setup to models which account for past order flow imbalances \cite{schneider2017,Alfonsi2016MultivariateFunctions,Benzaquen2017DissectingAnalysis} where price dynamics are
$$
\bp_t = \sum_{s \leq t} G(t-s) \bq_s + \bm \xi_t \, ,
$$
where $G \colon t \mapsto G(t) \in \mat{\nassets}$ captures the dependence on past order flow and $\bm \xi_t$ is a vector of zero-mean random variables. Then
\begin{align*}
\bdp_t &= G(0) \bq_{t} + \bmeta_t + \sum_{s < t} \left( G(t+\Delta t - s) - G(t-s) \right) \bq_s \, ,
\end{align*}
where $\bmeta_t := \bm \xi_{t+\Delta t} - \bm \xi_t$ is a vector of zero-mean random variables. \cref{ass:linear} ignores the last term of the right hand side in the above equation. Therefore, we can measure the validity of our approximation by comparing $\Gc_{ij} := \sum_{s < t} G_{ij}(t+\Delta t - s) - G_{ij}(t-s)$ and $G_{ij}(0)$, say with $\kappa_{ij} := \frac{\Gc_{ij}}{G_{ij}(0)}$.
\\ \\
For $\Delta t = \text{5 minutes}$ and on stocks, Figure 3 of \cite{Benzaquen2017DissectingAnalysis} shows $\kappa_{ij} \approx 20 \%$. Figure 5 of \cite{schneider2017} works in transaction time on bonds but a rough estimate for $\Delta t = \text{1 minute}$ yields $\kappa_{ij} \approx 30 \%$. This indicates our setup is relevant to capture the salient features of cross-impact.
\end{remark}

Before proceeding, we comment on the structure of the trading costs in our setup.

\begin{remark}[Trading costs]
\cref{eq:lin_impact} gives a prediction of portfolio trading costs. In particular, if one assumes that the difference between the arrival price and the execution price is given by $\bdp_t$, the cost incurred after the execution of the portfolio $\boldm{\xi}$ is
\begin{equation}
    \label{eq:cost_portfolio}
    \mathcal{C}(\boldm{\xi}) = \boldm{\xi}^\top \bdp_t  = \boldm{\xi}^\top \Lambda_t \boldm{\xi} + \boldm{\xi}^\top \Lambda_t \bar \bq_t + \boldm{\xi}^\top \bmeta_t \, ,
\end{equation}
where $\bar \bq_t$ is the order flow imbalance due to trades of other market participants. Thus 
$$
\E[\mathcal{C}(\boldm{\xi})] = \boldm{\xi}^\top \Lambda_t \boldm{\xi} + \boldm{\xi}^\top \Lambda_t \E[\bar \bq_t \mid \boldm{\xi}] + \boldm{\xi}^\top \bmeta_t \, ,
$$
where $\E[\bar \bq_t \mid \boldm{\xi}]$ represents the alignment of the market trades' conditioned to the traded portfolio. This may be non-zero because of herding, where our trades cause other investors to trade. The treatment of this term depends on the trading strategy and is outside of the scope of this paper. Thus, we assume that $\E[\bar \bq_t \mid \boldm{\xi}] = \E[\bar \bq_t] = 0$ so that the average impact costs of trading the portfolio $\boldm{\xi}$ in our setting is $\E[\mathcal{C}(\boldm{\xi})] = \boldm{\xi}^\top \Lambda_t \boldm{\xi}$.
\end{remark}

The next section examines desirable properties we want cross-impact models to satisfy.

\section{Axioms: the desirable properties of cross-impact models}
\label{sec:axioms}

To find a proper cross-impact model for \cref{eq:lin_impact}, we have two potentially conflicting objectives. The first is goodness of fit, or how well our model captures the influence of order flow to explain price changes. Cross-impact models are often selected on this basis alone \cite{Benzaquen2017DissectingAnalysis,schneider2017,wang2016cross}. This may yield good empirical fit but these models neglect another important aspect: the theoretical implications of our market model of \cref{eq:lin_impact}. For example, does our choice of cross-impact model imply that agents can abuse cross-impact to make a profit on average? These properties are critical if we want to use cross-impact models in applications, such as \cite{ekren2018optimal,ekren2019portfolio,garleanu2016dynamic,tsoukalas2019dynamic}.
\\ \\
To select a model on the basis of its implications, we need to establish which properties we would like a cross-impact model to satisfy. To address this issue, this section proposes a list of desirable properties of cross-impact models, which we dub \textit{axioms}. Each axiom translates a desired behaviour of cross-impact model, grounded in the implications of the cross-impact model for the evolution of prices.

\subsection{Symmetry}
The first type of axioms we introduce are \textit{symmetry} axioms. They ensure that the cross-impact is well-behaved under financially-grounded transformations of its variables. First, the cross-impact model should adapt to a re-ordering of the assets. This yields the following axiom.
\begin{axiom}[Permutational invariance]
\label{Axiom:permutational_invariance}
A cross-impact model $\Lambda$ is \emph{permutation-invariant} if, for any permutation matrix $P$ and $(\Sigma, \Omega, R) \in (\spd{n} \times \pd{n} \times \mat{n})$, 
\begin{equation*}
    \Lambda(P \Sigma P^\top, P \Omega P^\top, P R P^\top) = P \Lambda(\Sigma, \Omega, R) P^\top .
\end{equation*}
\end{axiom}
In the absence of any return or order flow correlation among assets, we expect price changes to be independent. The cross-impact model should then respect this property. This motivates the following axiom.
\begin{axiom}[Direct invariance]
\label{Axiom:direct_invariance}
A cross-impact model $\Lambda$ is \emph{direct-invariant} if, for any $\boldm{\sigma}, \boldm{\omega} \in \mathbb{R}^{n}_{+}$ $\boldm{r} \in \mathbb{R}^n$, 
\begin{equation*}
    \Lambda(\diag{\boldm{\sigma}}^2, \diag{\boldm{\omega}}^2, \diag{\boldm{r}}) = \sum_{i=1}^{n} \Lambda(\sigma_i^2 \be_i \be_i^\top , \omega_i^2 \be_i \be_i^\top, r_i \be_i\be_i^\top) \, ,
\end{equation*}
where $\be_i$ is the $i$-th element of the canonical basis.
\end{axiom}
Impact is expressed in a choice of currency units. However, the chosen currency should not matter and cross-impact models should adapt accordingly. The next axiom translates this property.
\begin{axiom}[Cash invariance]
\label{Axiom:cash_invariance}
A cross-impact model $\Lambda$ is \emph{cash-invariant} if, for any $\alpha > 0$, and $(\Sigma, \Omega, R) \in (\spd{n} \times \pd{n} \times \mat{n})$, 
\begin{equation*}
    \Lambda(\alpha^2 \Sigma, \Omega, \alpha R) = \alpha \Lambda(\Sigma, \Omega, R).
\end{equation*}
\end{axiom}
Similarly, cross-impact models should account for changes in volume units. For example, stock splits can double the number of outstanding shares and halve their values (if we ignore microstructural effects such as tick size and lot rounding). This leads to the following axiom.
\begin{axiom}[Split invariance]
\label{Axiom:split_invariance}
A cross-impact model $\Lambda$ is \emph{split-invariant} if, for any diagonal matrix of positive elements $D \in \mat{n}$ and $(\Sigma, \Omega, R) \in (\spd{n} \times \pd{n} \times \mat{n})$, 
$$
\Lambda(D^{-1} \Sigma D^{-1}, D \Omega D, D^{-1} R D) = D^{-1} \Lambda(\Sigma, \Omega, R) D^{-1}.
$$
\end{axiom}
The profit and loss of traders is invariant under orthogonal transformations (see \cref{eq:cost_portfolio}). It is natural to look for cross-impact models that share this property. As before, this ignores microstructural effects such as exchange trading fees, bid-ask spreads, etc. The following axiom introduces this property.
\begin{axiom}[Rotational invariance]
\label{Axiom:invariance_rotation}
A cross-impact model $\Lambda$ is \emph{rotation-invariant} if, for any real orthogonal matrix $O \in \orth{n}$ and $(\Sigma, \Omega, R) \in (\spd{n} \times \pd{n} \times \mat{n})$,
$$
\Lambda(O \Sigma O^\top, O \Omega O^\top, O R O^\top) = O \Lambda(\Sigma, \Omega, R) O^\top.
$$
\end{axiom}
We say of a model which does not satisfy \cref{Axiom:invariance_rotation} that it has a privileged basis. Note that any cross-impact model which satisfies \cref{Axiom:invariance_rotation,Axiom:split_invariance} is invariant under the action of any non-singular matrix $M$.
\\ \\
Among symmetry axioms, we expect permutational, direct and cash invariance (\cref{Axiom:permutational_invariance,Axiom:cash_invariance}) to be of critical importance as models which do not respect them would behave oddly. Split invariance ~(\cref{Axiom:split_invariance}) is realistic but it may break on small timescales due to microstructural effects. On the other hand, rotational invariance (\cref{Axiom:invariance_rotation}) is less plausible because markets have transaction costs, leverage constraints and other effects which break this symmetry.

\subsection{Arbitrage}
This family of axioms clarifies what properties a cross-impact model should satisfy to excludes any statistical arbitrage in the sense of~\cite{Gatheral2009No-Dynamic-ArbitrageImpact}, i.e. round-trip trading strategies with positive average profit. The first axiom involves \textit{static} arbitrages: single period trading strategies with average negative costs.
\begin{axiom}[Positive semi-definiteness]
\label{Axiom:positive_semidefinite}
The cross-impact model $\Lambda$ takes values in the space of positive semi-definite matrices. 
\end{axiom}
The next axiom involves \emph{dynamic} arbitrages, i.e. multi-period trading strategies in the spirit of~\cite{Alfonsi2016MultivariateFunctions, Gatheral2009No-Dynamic-ArbitrageImpact,schneider2017}. Even though these arbitrages cannot be exploited in our single-period setup, they would emerge by generalizing our setup to the multi-period setting as shown in~\cite{schneider2017}. This is why we choose to also consider this class of arbitrages.
\begin{axiom}[Symmetry]
\label{Axiom:symmetric}
The cross-impact model $\Lambda$ takes values in the space of symmetric matrices. 
\end{axiom}
\cref{Axiom:symmetric,Axiom:positive_semidefinite} together are sufficient to guarantee absence of statistical arbitrages. Arbitrage-related axioms are of great important in applications where the presence of arbitrages leads to odd behaviour. For example, \cite{Alfonsi2016MultivariateFunctions} highlights how arbitrageable cross-impact models lead to ill-behaved optimal trading strategies.
Although outside the scope of this paper, it is interesting to assess if real markets admit some kind of statistical arbitrage, and whether these hold when factoring other transaction costs (see \cite{schneider2017}).

\subsection{Fragmentation}
While the previous axioms ruled out statistical arbitrage, another related issue is what happens when trading assets (or combination of assets) which have constant prices. For example, consider a stock traded on multiple markets (say, Apple traded on the Nasdaq and on the Bats venues). For a reasonably large interval of time $\Delta t$ (and abstracting microstructural effects), we expect $p_{\textrm{Nasdaq}} - p_{\textrm{Bats}} = 0$. Thus, buying a volume $q = q_{\textrm{Nasdaq}} + q_{\textrm{Bats}}$ of Apple stock should yield the same cost no matter how one fragments the $q_{\textrm{Nasdaq}}$ units bought on Nasdaq and the $q_{\textrm{Bats}}$ units bought on Bats. For this reason, this axiom is dubbed \textit{fragmentation invariance}. 
\\ \\
We distinguish between three different forms of fragmentation invariance. The first, \emph{weak fragmentation invariance}, concerns the price changes given by a cross-impact model and is detailed in the next Axiom.
\begin{axiom}[Weak fragmentation invariance]
  \label{Axiom:weak_fragmentation_invariance}
  A cross-impact model $\Lambda$ is \emph{weakly fragmentation invariant} if, for any $(\Sigma, \Omega, R) \in (\spd{n} \times \pd{n} \times \mat{n})$ and $\emptyset \subset V \subseteq \ker{\Sigma}$,
$$
\Pi_{V} \Lambda(\Sigma, \Omega, R) = 0 \, ,
$$
where we recall that $\Pi_V$ denotes the projector on the linear subspace $V$. 
\end{axiom}
In practice, if the price of a linear combination of assets is constant, weak fragmentation invariance guarantees that its price cannot be moved through trading.

\begin{remark}
\label{remark:kernel_sigma_r}
From now on, we will implicitly assume that $\ker(\Sigma) \subseteq \ker(R^\top)$, which is consistent with the interpretation of \,$\Sigma$ and $R$ as covariations of prices and order flows. This implies that from the point of view of the fragmentation-related axioms, any condition involving the the kernel of \,$\Sigma$ will be naturally related to the kernel of $R^\top$ as well.
\end{remark}

We obtain a stronger condition if we require volume traded in zero-volatility directions to induce no price impact. This leads to the following Axiom.

\begin{axiom}[Semi-Strong fragmentation invariance]
\label{Axiom:semi_strong_fragmentation_invariance}
A cross-impact model satisfies \emph{semi-strong fragmentation invariance} if, besides satisfying the weak fragmentation invariance \cref{Axiom:weak_fragmentation_invariance}, for any $(\Sigma, \Omega, R) \in (\spd{n} \times \pd{n} \times \mat{n})$ and $\emptyset \subset V \subseteq \ker{\Sigma}$,
$$
\Lambda(\Sigma, \Omega, R) \Pi_{V} = 0 \, .
$$
\end{axiom}

We can go one step further by ensuring that the cross-impact model itself should also not depend on how zero-volatility directions are traded by \emph{other} market members. This is \emph{strong fragmentation invariance}, the subject of the next Axiom.

\begin{axiom}[Strong fragmentation invariance]
\label{Axiom:strong_fragmentation_invariance}
A cross-impact model $\Lambda$ is \emph{strongly fragmentation invariant} if, besides satisfying semi-strong fragmentation invariance (\cref{Axiom:semi_strong_fragmentation_invariance}),  for any $(\Sigma, \Omega, R) \in (\spd{n} \times \pd{n} \times \mat{n})$ and $\emptyset \subset V \subseteq \ker(\Sigma)$,
$$
\Lambda(\Sigma, \Omega, R) = \Lambda(\bar \Pi_V \Sigma \bar \Pi_V,  \bar \Pi_{V} \Omega  \bar \Pi_{V}, \bar \Pi_V R \bar \Pi_{V}) \, .
$$
\end{axiom}

Weak fragmentation invariance (\cref{Axiom:weak_fragmentation_invariance}) is critical since it prevents models from predicting price changes for zero-volatility instruments. Furthermore, it properly aggregates liquidity of an asset traded on multiple venues. For the same reasons, semi-strong and strong fragmentation invariance (\cref{Axiom:semi_strong_fragmentation_invariance,Axiom:strong_fragmentation_invariance})  should be of crucial importance.

\subsection{Stability}
Fragmentation invariance axioms constrain cross-impact models in extreme regimes of price correlations. Similarly, stability axioms control behaviour in extremes of liquidity. Intuitively, price manipulation of liquid products using illiquid instruments should be excluded. 
\\ \\
We model this by defining a set $V$ of illiquid instruments. We consider the matrix $\bar \Pi_V + \varepsilon \Pi_V$ that multiplies by $\varepsilon \ll 1$ the liquidity of all instruments belonging to $V$. After multiplying the traded order flow by this matrix, the observables become
\begin{align*}
  \Sigmaqeps &:= \Sigma \\
  \Omegaqeps &:= (\bar \Pi_V + \varepsilon \Pi_V) \Omega (\bar \Pi_V + \varepsilon \Pi_V) \\
  \Rqeps &:= R (\bar \Pi_V + \varepsilon \Pi_V) \, .
\end{align*}
We are now ready to formulate liquidity axioms. First, trading illiquid instruments should not lead to large impact on liquid instruments. We would otherwise be able to manipulate the prices of liquid instruments. The converse should be true: we should not be able to manipulate prices of illiquid instruments by trading liquid instruments. This motivates the next axiom.

\begin{axiom}[Weak Cross-Stability]
\label{Axiom:weak_cross_stability}
A cross-impact model $\Lambda$ is \emph{weakly cross-stable} if, for any $(\Sigma, \Omega, R) \in (\spd{n} \times \pd{n} \times \mat{n})$ and linear subspace $V$ and using the above notations,
\begin{align}
\bar{\Pi}_V  \Lambda(\Sigmaqeps, \Omegaqeps, \Rqeps) \Pi_V &\underset{\varepsilon \to 0}{=} O(1) \\
\Pi_V  \Lambda(\Sigmaqeps, \Omegaqeps, \Rqeps) \bar \Pi_V &\underset{\varepsilon \to 0}{=} O(1) \, .
\end{align}
\end{axiom}

We can formulate a stronger cross-stability property. The next axiom formalizes the intuition that impact among liquid assets should be independent of the behavior of illiquid assets.
\begin{axiom}[Strong Cross-Stability]
\label{Axiom:strong_cross_stability}
A cross-impact model $\Lambda$ is \emph{strongly cross-stable} if, in addition to satisfying weak-cross stability (~\cref{Axiom:weak_cross_stability}), for any $(\Sigma, \Omega, R) \in (\spd{n} \times \pd{n} \times \mat{n})$ and linear subspace $V$ and using the above notations,
$$
\bar{\Pi}_V  \Lambda\left(\Sigmaqeps, \Omegaqeps, \Rqeps\right) \bar \Pi_V \underset{\varepsilon \to 0}{\to} \bar{\Pi}_V  \Lambda\left(\bar \Pi_V \Sigma \bar \Pi_V, \bar \Pi_V \Omega \bar \Pi_V, \bar \Pi_V R \bar \Pi_V\right) \bar \Pi_V
$$
\end{axiom}
An unresolved question is the effect of trading illiquid instruments on illiquid products. The following axiom deals with this issue.
\begin{axiom}[Self-Stability]
  \label{Axiom:self-stability}
  A cross-impact model is \emph{self-stable} if, for any $(\Sigma, \Omega, R) \in (\spd{n} \times \pd{n} \times \mat{n})$, subspace $V$ and using the above notations,
  \begin{equation}
    \Pi_V  \Lambda(\Sigmaqeps, \Omegaqeps, \Rqeps) \Pi_V \underset{\varepsilon \to 0}{=} O(1).
  \end{equation}
\end{axiom}
Intuitively we want to avoid this property since it indicates that, even though a product is illiquid ($q \propto \varepsilon$, so that one would expect a diverging impact) the predicted cost of trading such product can be finite. 
\\ \\
We believe weak cross-stability (\cref{Axiom:weak_cross_stability}) is fundamental. Indeed, it should be impossible to manipulate prices from liquid assets by trading illiquid assets and vice-versa. We would also like the stronger version of this axiom (\cref{Axiom:strong_cross_stability}) to hold: liquid instruments should be insensitive to trading on illiquid ones. On the other hand, self-stability (\cref{Axiom:self-stability}) does not penalize trading illiquid instruments. Thus, it is undesirable in applications.

\subsection{Predicted covariance}
Finally, it can be interesting to consider whether a cross-impact model predicts a contribution to the return covariance that is proportional to $\Sigma$ or not. 
\begin{axiom}[Return covariance consistency]
\label{Axiom:consistency_correlation}
A cross-impact model $\Lambda$ is \emph{return covariance consistent} if, for any $(\Sigma, \Omega, R) \in (\spd{n} \times \pd{n} \times \mat{n})$, it satisfies (up to a multiplicative constant):
\begin{equation*}
    \Sigma = \Lambda(\Sigma, \Omega, R) \Omega  \Lambda(\Sigma, \Omega, R)^\top.
  \end{equation*}
\end{axiom}
This axiom is motivated by the fact that under the model in \cref{eq:lin_impact}, we expect return covariances to be given by
$$
\Sigma = \E[\boldm{\Delta p} \boldm{\Delta p}^\top] = \Lambda \Omega \Lambda^\top + \E[\boldm{\eta} \boldm{\eta}^\top] \, ,
$$
so that if one assumes that the fundamental return covariance is proportional to the predicted one, i.e. $\E[\boldm{\eta} \boldm{\eta}^\top] \propto \Sigma$, one would recover return covariance consistency. This property controls the predicted price changes of the model, but we have no strong reason to believe cross-impact models should satisfy it.

\subsection{Link between axioms}

Fragmentation and cross-stability are related for split and rotation-invariant cross-impact models. The next proposition shows that fragmentation invariance implies cross-stability properties for continuous cross-impact models.
\begin{proposition}
\label{prop:fragm_to_liq_main}
Let $\Lambda$ be a jointly continuous cross-impact model which satisfies split and rotational invariance~(\cref{Axiom:invariance_rotation,Axiom:split_invariance}). Then
\begin{enumerate}[(i)]
    \item If $\Lambda$ satisfies semi-strong fragmentation invariance~(\cref{Axiom:semi_strong_fragmentation_invariance}), then it is weakly cross-stable (\cref{Axiom:weak_cross_stability}).
    \item If $\Lambda$ is strongly fragmentation invariant (\cref{Axiom:strong_fragmentation_invariance}), then it is strongly cross-stable (\cref{Axiom:strong_cross_stability}).
\end{enumerate}
\end{proposition}
We prove \cref{prop:fragm_to_liq_main} in \cref{sec:fragm_liq}. While the converse is not true, the next proposition shows that, given an additional regularity condition, cross-stability implies fragmentation invariance.
\begin{proposition}
\label{prop:liq_to_frag_main}
Let $\Lambda$ be a jointly continuous cross-impact model which satisfies split and rotational invariance~(\cref{Axiom:invariance_rotation,Axiom:split_invariance}). We further assume that, for any linear subspace $V$ and using the notations of the previous section, $\varepsilon^{2} \Lambda(\Sigmaqeps, \Omegaqeps, \Rqeps)  \underset{\varepsilon \to 0}{\to} 0$. Then
\begin{enumerate}[(i)]
    \item If $\Lambda$ is weakly cross-stable (\cref{Axiom:weak_cross_stability}), then it satisfies semi-strong fragmentation invariance~(\cref{Axiom:semi_strong_fragmentation_invariance}).
    \item If $\Lambda$ is strongly cross-stable (\cref{Axiom:strong_cross_stability}), then it is strongly fragmentation invariant (\cref{Axiom:strong_fragmentation_invariance}). 
\end{enumerate}
\end{proposition}
We prove \cref{prop:liq_to_frag_main} in \cref{sec:fragm_liq}. A particularly interesting result of \cref{prop:liq_to_frag_main,prop:fragm_to_liq_main} is that for continuous cross-impact models which satisfy the regularity property of \cref{prop:liq_to_frag_main}, fragmentation invariance and cross-stability are equivalent.

\section{Candidate cross-impact models}
\label{sec:candidate_cross_impact_models}

Now that we have characterized the desirable properties of cross-impact models, we provide a set of cross-impact models and detail which axioms they satisfy. Their empirical performance will be assessed in \cref{sec:illustrative_examples_empirical}. We divide these models in two classes; those that are based on the return covariance $\Sigma$ and those based on the response $R$.
\\ \\
Before presenting the different cross-impact models, we introduce some notation. For convenience, we will note the price volatility $\boldm{\sigma} := (\sqrt{\Sigma_{11}}, \cdots, \sqrt{\Sigma_{\nassets \nassets}})$, the signed order flow volatility $\boldm{\omega} := (\sqrt{\Omega_{11}}, \cdots, \sqrt{\Omega_{\nassets \nassets}})$, and the price and flow correlations $\rho := \diag{\boldm{\sigma}}^{-1} \Sigma \diag{\boldm{\sigma}}^{-1}$, $\rho_{\Omega} := \diag{\boldm{\omega}}^{-1} \Omega \diag{\boldm{\omega}}^{-1}$.

\subsection{Return covariance based models}
\label{sec:price-covar-based}

Let us start with the simplest possible linear impact model:  one without cross-impact.
\begin{Definition}[\texttt{direct} model]
The \emph{\texttt{direct}} model is defined for any $(\Sigma, \Omega, R) \in (\spd{n} \times \pd{n} \times \mat{n})$ as
\begin{equation}\label{eq:pdir}
    \Lambda_{\emph{\pdir}}(\Sigma, \Omega, R) := \diag{\boldm{\sigma}}^{1/2} \diag{\boldm{\omega}}^{-1/2}.
\end{equation}
\end{Definition}
 To  generalize this model to the multivariate setting while respecting cash invariance, weak fragmentation invariance and consistency with correlations, a first idea is to use the matrices $\Sigma^{1/2}$ and $\Omega^{-1/2}$. Since $\Omega^{-1/2} \boldm{q}$ is a whitening transformation, this model is referred to as the \texttt{whitening} model.
\begin{Definition}[\texttt{whitening} model]
Recall that given $M \in \spd{n}$, $M^{1/2}$ indicates a symmetric matrix factorization (i.e., $M^{1/2} (M^{1/2})^\top = \id$). The \emph{\texttt{whitening}} model\footnote{The \texttt{whitening} model is not independent of the symmetric factorization chosen for $\Sigma$ and $\Omega$. As convention, we will take the square root obtained by an orthogonal decomposition of each matrix and the square root of their eigenvalues.} is defined, for any $(\Sigma, \Omega, R) \in (\spd{n} \times \pd{n} \times \mat{n})$, as
\begin{equation}
    \label{eq:pwhi}
    \Lambda_{\emph{\pwhi}}(\Sigma, \Omega, R) := \Sigma^{1/2} \Omega^{-1/2}.
\end{equation}
\end{Definition}
Unfortunately, this model does not respect symmetry, positive-definiteness, strong fragmentation invariance or weak cross-stability (\cref{Axiom:positive_semidefinite,Axiom:symmetric,Axiom:strong_fragmentation_invariance,Axiom:weak_cross_stability}). To impose symmetry and strong fragmentation invariance, the  \texttt{el} model\footnote{The model proposed in~\cite{MBE2017} is actually the response-based one, referred later as \texttt{r-el$\star$} model.} proposed in~\cite{MBE2017} is directly expressed in the basis of the return covariance matrix.

\begin{Definition}[\pelm{} model]
The eigenliquidity (\emph{\pelm{}}) model is defined, for any $(\Sigma, \Omega, R) \in (\spd{n} \times \pd{n} \times \mat{n})$, as
\begin{equation}
\Lambda_{\emph{\pelm{}}}(\Sigma, \Omega, R) := \sum_{a=1}^n \boldm{s}_a \dfrac{\sqrt{\lambda_a}}{(\boldm{s}_a^\top \Omega \boldm{s}_a)^{1/2}} \boldm{s}_a^\top,
\end{equation}
where we have introduced the eigenvalue decomposition of $\Sigma = \sum_{a=1}^n \boldm{s}_a \lambda_a \boldm{s}_a^\top$.
\end{Definition}
The \pelm{} model is cross-stable, self-stable (\cref{Axiom:strong_cross_stability,Axiom:weak_cross_stability,Axiom:self-stability}) and is return covariance inconsistent (\cref{Axiom:consistency_correlation}). As mentioned above, there is in fact only one model which satisfies all the axioms that we have provided: the so-called multivariate \pkyle{} model, see~\cite{delMolino2018TheDifferent,Caballe1994ImperfectNeutrality}.

\begin{Definition}[\pkyle{} model]
The \emph{\pkyle{}} model is defined, for any $(\Sigma, \Omega, R) \in (\spd{n} \times \pd{n} \times \mat{n})$, as
  \begin{equation}
    \Lambda_{\emph{\pkyle{}}}(\Sigma, \Omega, R) := (\Omega^{-1/2})^\top \sqrt{(\Omega^{1/2})^\top \Sigma \Omega^{1/2}} \Omega^{-1/2}.
  \end{equation}

\end{Definition}

The \pkyle{} model plays a fundamental role as it is the only model which satisfies all axioms. This may explain why it appears in many different settings, seemingly unrelated to the Kyle insider trading setup \cite{Gueant2017OptimalMaking, Evangelista2018NewMaking,rosenbaum2021characterisation}. The next proposition shows it is the only model which satisfies arbitrage axioms and return covariance consistency.

\begin{proposition}
\label{prop:kyle_arb_cov_unique_main}
Let $\Lambda$ be a symmetric, positive-semidefinite and return covariance consistent cross-impact model (\cref{Axiom:symmetric,Axiom:positive_semidefinite,Axiom:consistency_correlation}).
Then $\Lambda = \Lambda_{\emph{\pkyle}}$ up to a multiplicative constant.
\end{proposition}

The proof of \cref{prop:kyle_arb_cov_unique_main} is given in \cref{sec:connections}. The next proposition further shows that the \pkyle{} model is also the only return covariance based model which satisfies all symmetry axioms.

\begin{proposition}
\label{prop:kyle_sym_unique_main}
A return covariance based cross-impact model $\Lambda$ that is both split-invariant and rotation-invariant (\cref{Axiom:invariance_rotation,Axiom:split_invariance}) can always be written in the form
	$$
	\Lambda(\Sigma, \Omega) = \Lc^{-\top} U F(\bmu) U^\top \Lc^{-1},
	$$
	where
		\begin{align*}
	    \Omega = \Lc \Lc^\top \quad ; \quad 
	    \hat \Sigma := \Lc^{\top} \Sigma \Lc \quad ; \quad 
	    U^\top \hat{\Sigma} U := \diag{\bmu} \quad ; \quad 
	    F(\bmu) := \Lambda(\diag{\bmu}, \id).
	\end{align*}
Furthermore, if $\Lambda$ is cash-invariant and direct-invariant (\cref{Axiom:cash_invariance,Axiom:direct_invariance}), then $F(\boldm{\mu}) = \diag{\boldm{\mu}}^{1/2}$ up to a multiplicative constant and $\Lambda = \Lambda_{\emph{\pkyle}}$ up to a multiplicative constant.
\end{proposition}

The proof of \cref{prop:kyle_sym_unique_main} is given in \cref{sec:connections}. 

\subsection{Response based models}
\label{sec:response-based}
All the models presented above assume that it is possible to relate the effect of the order flow imbalance solely with the return  and order flow covariances. This section examines models which also use the response matrix $R$. First, we can define a response-based direct impact model similar to \cref{eq:pdir}.
\begin{Definition}[\rdir{} model]
The response direct (\emph{\rdir{}}) model is defined, for any $(\Sigma, \Omega, R) \in (\spd{n} \times \pd{n} \times \mat{n})$, as
\begin{equation*}
    \Lambda_{\emph{\rdir}}(\Sigma, \Omega, R) := \diag{R_{11}, \cdots, R_{\nassets \nassets}} \diag{\boldm{\omega}}^{-1} \,.
\end{equation*}
\end{Definition}
This model corresponds to the maximum likelihood estimator of the cross-impact matrix $\Lambda$ under the constraint $\Lambda_{ij} = 0$ for $i \neq j$. Removing this constraint, one obtains the multiavariate maximum likelihood estimator defined below.
\begin{Definition}[\rmle{} model]
The maximum likelihood (\emph{\rmle{}}) model is defined, for any  $(\Sigma, \Omega, R)$ in $ (\spd{n} \times \pd{n} \times \mat{n})$, as
$$
\Lambda_{\emph{\rmle{}}}(\Sigma, \Omega, R) := R \Omega^{-1}.
$$
\end{Definition}
The \rmle{} does not satisfy desirable arbitrage or liquidity axioms. Thus, for similar reasons the \pelm{} was introduced, we introduce a \relm{} model, so to have a response-based model satisfying more axioms while coinciding with the \rmle{} model when $R$ and $\Omega$ commute.
\begin{Definition}[\relm{} model]
The response-based eigenliquidity (\emph{\relm{}}) model is defined, for any  $(\Sigma, \Omega, R)$ in $ (\spd{n} \times \pd{n} \times \mat{n})$, as
\begin{equation}
  \label{eq:relm}
  \Lambda_{\emph{\relm{}}}(\Sigma, \Omega, R) := \sum_a \boldm{s}_a \dfrac{\boldm{s}_a^\top R \boldm{s_a}}{\boldm{s}_a^\top \Omega \boldm{s}_a} \boldm{s}_a^\top \, ,  
\end{equation}
where $\boldm{s_a}$ are the eigenvectors of $\Sigma$.
\end{Definition}
Finally, we can replicate the construction of the \pkyle{} estimator in a response-based context to obtain the following model.
\begin{Definition}[\rkyle{} model]
The response-based Kyle (\emph{\rkyle{}}) model is defined, for any  $(\Sigma, \Omega, R)$ in $ (\spd{n} \times \pd{n} \times \mat{n})$, as
  \begin{equation}
    \label{eq:rkyle}
    \Lambda_{\emph{\rkyle{}}}(\Sigma, \Omega, R) := (\Omega^{-1/2})^\top \sqrt{(\Omega^{1/2})^\top R \Omega^{-1} R^\top \Omega^{1/2}} \Omega^{-1/2}.
  \end{equation}
\end{Definition}

\subsection{The $\star$ transformation}
\label{sec:-transformation}

Some of the models defined in the previous section (\pwhi{}, \pelm{}, \relm{}) violate split invariance even though they are well-behaved under rotation. We can trade one for the other through the following transformation.

\begin{Definition}[The $\star$ transformation]
Given a cross-impact model $\Lambda$, the starred version of $\Lambda$, written $\Lambda^{\star}$, is a cross-impact model defined for any  $(\Sigma, \Omega, R)$ in $ (\spd{n} \times \pd{n} \times \mat{n})$ as
$$
\Lambda^{\star}(\Sigma, \Omega, R) := \diag{\boldm{\sigma}} \Lambda ( \rho, \Omega^{\star} , R^{\star} )\diag{\boldm{\sigma}} \, ,
$$
where we have defined  $\Omega^{\star} =  \diag{\boldm{\sigma}} \Omega  \diag{\boldm{\sigma}}$ and $R^{\star} =  \diag{\boldm{\sigma}}^{-1} R  \diag{\boldm{\sigma}}$.
\end{Definition}
In practice, the starred version of a cross-impact model applies the original cross-impact model after rescaling all the observables in units of risk via a multiplication by the volatility $\boldm{\sigma}$. Naturally, this transformation has no effect on models that satisfy split invariance.

\subsection{Axioms satisfied by each model}

\cref{table:axioms_models} summarises the axioms satisfied by each model. Most results are straightforward and omitted for conciseness. We include some of the slightly less trivial proofs of the axioms satisfied by the \pkyle{} model in \cref{proof:kyle}. We summarise some of the connections between Axioms in \cref{table:axioms_relations}.

\begin{table}[H]
    \centering
    \begin{adjustwidth}{-1cm}{}
    \small
    \begin{tabular}{lcccccccccccccccccc} 
    \toprule
    \multicolumn{1}{c}{Model} & \multicolumn{5}{c}{Symmetries} & \phantom{a} & \multicolumn{2}{c}{Arbitrage} & \phantom{a} & \multicolumn{3}{c}{Fragmentation} & \phantom{a} & \multicolumn{3}{c}{Liquidity} & \phantom{a} & \multicolumn{1}{c}{Covariances} \\
    \cmidrule{2-6} \cmidrule{8-9} \cmidrule{11-13} \cmidrule{15-17} \cmidrule{19-19} & PI & DI & CI & SI & RI && SA & DA && WFI & SSFI & SFI && WCS & SCS & SS && PCC \\
    \texttt{direct}              & \gcmark & \gcmark & \gcmark &\gcmark & \yxmark && \gcmark & \gcmark && \rxmark & \rxmark & \rxmark && \gcmark & \gcmark & \gxmark && \yxmark \\
    \texttt{whitening}           & \gcmark & \gcmark & \gcmark & \rxmark & \gcmark && \rxmark & \rxmark && \gcmark & \rxmark & \rxmark && \rxmark & \rxmark & \gxmark && \gcmark \\
    \texttt{whitening$\star$}    & \gcmark   & \gcmark & \gcmark & \gcmark & \yxmark && \rxmark & \rxmark && \gcmark & \rxmark & \rxmark && \rxmark & \rxmark & \gxmark && \gcmark \\
    \texttt{el}      & \gcmark & \gcmark & \gcmark & \rxmark & \gcmark && \gcmark & \gcmark && \gcmark & \gcmark & \gcmark && \gcmark & \gcmark & \ycmark && \yxmark \\
    \texttt{el$\star$}    & \gcmark & \gcmark & \gcmark & \gcmark & \yxmark && \gcmark & \gcmark && \gcmark & \gcmark & \gcmark && \gcmark & \gcmark & \ycmark && \yxmark \\ 
    \texttt{kyle}             & \gcmark & \gcmark & \gcmark & \gcmark & \gcmark && \gcmark & \gcmark && \gcmark & \gcmark & \gcmark && \gcmark & \gcmark & \gxmark && \gcmark \\ 
    \texttt{r-direct}            &  \gcmark & \gcmark & \gcmark & \gcmark & \yxmark && \gcmark & \rxmark && \rxmark & \rxmark & \rxmark && \gcmark & \gcmark & \gxmark && \yxmark \\ 
    \texttt{ml}  & \gcmark & \gcmark & \gcmark & \gcmark & \gcmark && \rxmark & \rxmark && \gcmark & \rxmark & \rxmark && \rxmark & \rxmark & \gxmark && \yxmark  \\ 
    \texttt{r-el}    & \gcmark  & \gcmark & \gcmark & \rxmark & \gcmark && \rxmark & \gcmark && \gcmark & \gcmark & \gcmark && \gcmark & \gcmark & \ycmark && \yxmark  \\
    \texttt{r-el$\star$}   & \gcmark  & \gcmark &  \gcmark & \gcmark & \yxmark && \rxmark & \gcmark && \gcmark & \gcmark & \gcmark &&  \gcmark &\gcmark & \ycmark && \yxmark  \\ 
    \texttt{r-kyle}             & \gcmark  & \gcmark & \gcmark & \gcmark & \gcmark && \gcmark & \gcmark && \gcmark & \gcmark & \gcmark && \gcmark & \gcmark & \gxmark && \yxmark \\ 
    \bottomrule
    \end{tabular}
    \end{adjustwidth}
    \caption{\textbf{Summary of axioms satisfied by different cross-impact model.} \\
    We use the symbol \cmark \, for axioms that are satisfied and \xmark \, for axioms that are violated.
    We use the color \green{green} in order to label a desirable property of the model, \red{red} for an undesirable property of the model. \yellow{Yellow} is used for properties/models whose violation might not be particularly relevant in order to explain empirical data, although they are interesting to consider. Axioms are grouped by category and the order in which they were presented in the text.}
    \label{table:axioms_models}
\end{table}

\begin{table}[H]
    \centering
    \begin{adjustwidth}{-1cm}{}
    \small
    \begin{tabular}{lcccccccccccccccccc} 
    \toprule
    \multicolumn{1}{c}{Result} & \multicolumn{5}{c}{Symmetries} & \phantom{a} & \multicolumn{2}{c}{Arbitrage} & \phantom{a} & \multicolumn{3}{c}{Fragmentation} & \phantom{a} & \multicolumn{3}{c}{Liquidity} & \phantom{a} & \multicolumn{1}{c}{Covariances} \\
    \cmidrule{2-6} \cmidrule{8-9} \cmidrule{11-13} \cmidrule{15-17} \cmidrule{19-19} & PI & DI & CI & SI & RI && SA & DA && WFI & SSFI & SFI && WCS & SCS & SS && PCC \\
    \cref{prop:fragm_to_liq_main} (i)  & & & & H & H && & & & & H & && \gcmark & & & &  \\
    \cref{prop:fragm_to_liq_main} (ii)  & & & & H & H && & & & & & H && & \gcmark & & &  \\
    \cref{prop:liq_to_frag_main} (i)  & & & & H & H && & & & & \gcmark &  && H &  & & &  \\
    \cref{prop:liq_to_frag_main} (ii)  & & & & H & H && & & & & & \gcmark && & H &  & &  \\
    \cref{prop:kyle_arb_cov_unique_main} & \gcmark & \gcmark & \gcmark & \gcmark & \gcmark && H & H && \gcmark & \gcmark & \gcmark && \gcmark & \gcmark & \gxmark && H \\ 
    \cref{prop:kyle_sym_unique_main} & \gcmark & H & H & H & H && \gcmark & \gcmark && \gcmark & \gcmark & \gcmark && \gcmark & \gcmark & \gxmark && \gcmark \\ 
    \bottomrule
    \end{tabular}
    \end{adjustwidth}
    \caption{\textbf{Salient relations among the axioms introduced in the paper.} \\
    The table summarises the results of different propositions relating axioms together. For a given result, we use the symbol H to denote a condition that holds by hypothesis. On the same row, we mark satisfied axioms using the notation of Table \ref{table:axioms_models}.}
    \label{table:axioms_relations}
\end{table}

\section{Goodness-of-fit of cross-impact models}
\label{sec:illustrative_examples_empirical}

\cref{sec:candidate_cross_impact_models,sec:axioms} listed desirable properties of cross-impact models and examined which were satisfied by a variety of candidate models. This enabled us to understand the theoretical implications of a given cross-impact model. However, well-behaved models which poorly explain data are of little use. The goal of this section is to assess the goodness-of-fit of the cross-impact models listed in \cref{sec:candidate_cross_impact_models} to understand which models satisfy desirable properties and fit data well.

\subsection{Methodology}

To assess goodness-of-fit, we select the timescale $\Delta t$ to be one minute in order to avoid microstructural effects while being small. For a given cross-impact model $\Lambda$, the predicted price change for the time window $[t, t + \Delta t]$ due to the measured order flow imbalance $\bq_t$ on that time window is
$$
\widehat{\bdp}_t := \Lambda(\Sigma_t, \Omega_t, R_t) \bq_t \, ,
$$
where $\Sigma_t, \Omega_t, R_t$ are the covariances defined in \cref{ass:sufficient_stats}, which we will estimate using empirical data. 
\\ \\
To evaluate quality of fit of the cross-impact model $\Lambda$, we compare the predicted price changes $\widehat{\bdp}_t$ to the realised price changes $\bdp_t$, using three different indicators of performance which emphasize different aspects of  prediction errors. All  three indicators are parametrized by a symmetric, positive definite matrix $M \in \spd{n}$, $M \neq 0$. Given a realization of the price process $\{ \bdp_t \}_{t=1}^T$ of length $T$ and a corresponding series of predictions  $\{ \widehat{\bdp}_t\}_{t=1}^T$, the $M$-weighted generalized $\rsq$ is defined as
$$
\rsq(M) := 1 - \dfrac{\sum_{1 \leq t \leq T} (\boldm{\Delta p}_t - \boldm{\widehat{\Delta p}}_t)^\top M (\boldm{\Delta p}_t - \boldm{\widehat{\Delta p}}_t)}{\sum_{1 \leq t \leq T} \boldm{\Delta p}^\top_t M \boldm{\Delta p}_t} \, .
$$
The close the score is to one, the better the fit to empirical data. To highlight different sources of error, we consider the following choices of $M$:
\begin{enumerate}[(i)]
    \itemsep0em
    \item $M = I_{\sigma} := \diag{\boldm{\sigma}}^{-1}$, to account for errors relative to the typical deviation of the asset considered. This type of error is relevant for strategies predicting idiosyncratic moves of the constituents of the basket, rather than strategies betting on correlated market moves.
    \item $M = J_{\sigma} := (\Sigma_{ii}^{-1/2} \Sigma_{jj}^{-1/2})_{1 \leq i,j \leq m}$, to check if the model successfully forecasts the overall direction of all assets. This is relevant for strategies predicting global moves of the constituents of the basket.
    \item $M = \Sigma^{-1}$, to consider how well the model predicts the individual modes of the return covariance matrix. This would be the relevant error measure for strategies that place a constant amount of risk on the modes of the correlation matrix, leveraging up combinations of products with low volatility and scaling down market direction that exhibit large fluctuations. \footnote{Note that this measure strongly penalizes models violating fragmentation invariance: errors along modes of zero risk should \emph{a-priori} be enhanced by an infinite amount. In this study we have decided to clip the eigenvalues of $\Sigma$ to a small, non-zero amount equal to $10^{-15}$.}
    \end{enumerate}
Given $M \in \spd{n}$, $M \neq 0$, we compute scores on empirical data in the following manner. First, we divide data into two subsets of roughly equal length: data from 2016 on the one hand and in 2017 on the other hand. Given data from year $X$ and year $Y$, we calibrate estimators and cross-impact models on year $X$ and use models to predict price changes in year $Y$, writing $\rsq_{X \rightarrow Y}(M)$ for the average score. In-sample scores are defined as $\rsq_{\textnormal{in}}(M) := \frac{1}{2}(\rsq_{2016 \rightarrow 2016}(M) + \rsq_{2017 \rightarrow 2017}(M))$ while out-of-sample scores are defined as $\rsq_{\textnormal{out}}(M) := \frac{1}{2}(\rsq_{2016 \rightarrow 2017}(M) + \rsq_{2017 \rightarrow 2016}(M))$.

\subsection{Data used}

To assess goodness-of-fit in a variety of different conditions, we stress-test models on three different markets with different key characteristics. We detail each dataset here and the reason we chose them. Detailed descriptions of each dataset, the estimation of covariance matrices and of the cross-impact models is given in \cref{app:data}.
\\ \\
The first dataset comprises two NYMEX Crude Oil future contracts and the corresponding Calendar Spread contract. The first two contracts (respectively, CRUDE0 and CRUDE1) entail an agreement to buy or sell 1000 barrels of oil either at the next month or at the subsequent month. The Calendar Spread CRUDE1\_0 swaps the front month future with the contract settling on the following month. Because of the strong correlations among the two futures, the price of the calendar spread has very small fluctuations. This dataset allows us to test the importance of fragmentation axioms. Further details about this data are given in \cref{app:data:crude}.
\\ \\
While relevant to illustrate the importance of fragmentation invariance, the previous dataset corresponds to a pathological case where $\Sigma$ has only one large non-zero eigenvalue, so that cross-impact models give similar results. To circumvent this issue, we look at 10-year US Treasury note futures and E-Mini S\&P500 futures. We collect data from the Chicago Mercantile Exchange and use the first two upcoming maturities of both contracts (respectively called SPMINI and SPMINI3 for E-Mini S\&P500 futures and 10USNOTE and 10USNOTE3 for 10-year US treasury notes). Further details about this data are given in \cref{app:data:bonds}.
\\ \\
The previous datasets give us no clear conclusion on the role of stability axioms. In both examples illiquid assets were highly correlated to other liquid assets. This extreme regime of correlations makes it harder to analyse the role of liquidity. To circumvent this issue, we study the behavior of cross-impact models in the low-correlation, many assets regime, using stocks data. Further details about this data are given in \cref{app:data:stocks}.

\subsection{Goodness-of-fit}

The goodness-of-fit results for each model, dataset and score are presented in \cref{fig:gof} in \cref{app:goodness_of_fit}. Overall, on all datasets, the cross-impact models with the best goodness-of-fit are the \relm{}, \pkyle{} and \rmle{} models. They significantly outperform models which do not account for cross-sectional effects, such as the \rdir{} model. In high-correlation regimes, such as on the Crude and Bonds and Indices datasets, this gap is more pronounced. Among these models, it is remarkable that the \pkyle{} model satisfies all axioms and achieves comparable performance to the \rmle{} model, which maximises empirical fit but has issues related to arbitrage.
\\ \\
Given these results, we focus on analysing the influence of certain market parameters on the goodness-of-fit of cross-impact models. The next section examines the role of the liquidity.

\subsection{Goodness-of-fit relative to liquidity}

\begin{figure}[t!]
	\centering
		\centering
		\includegraphics[width=0.58\linewidth]{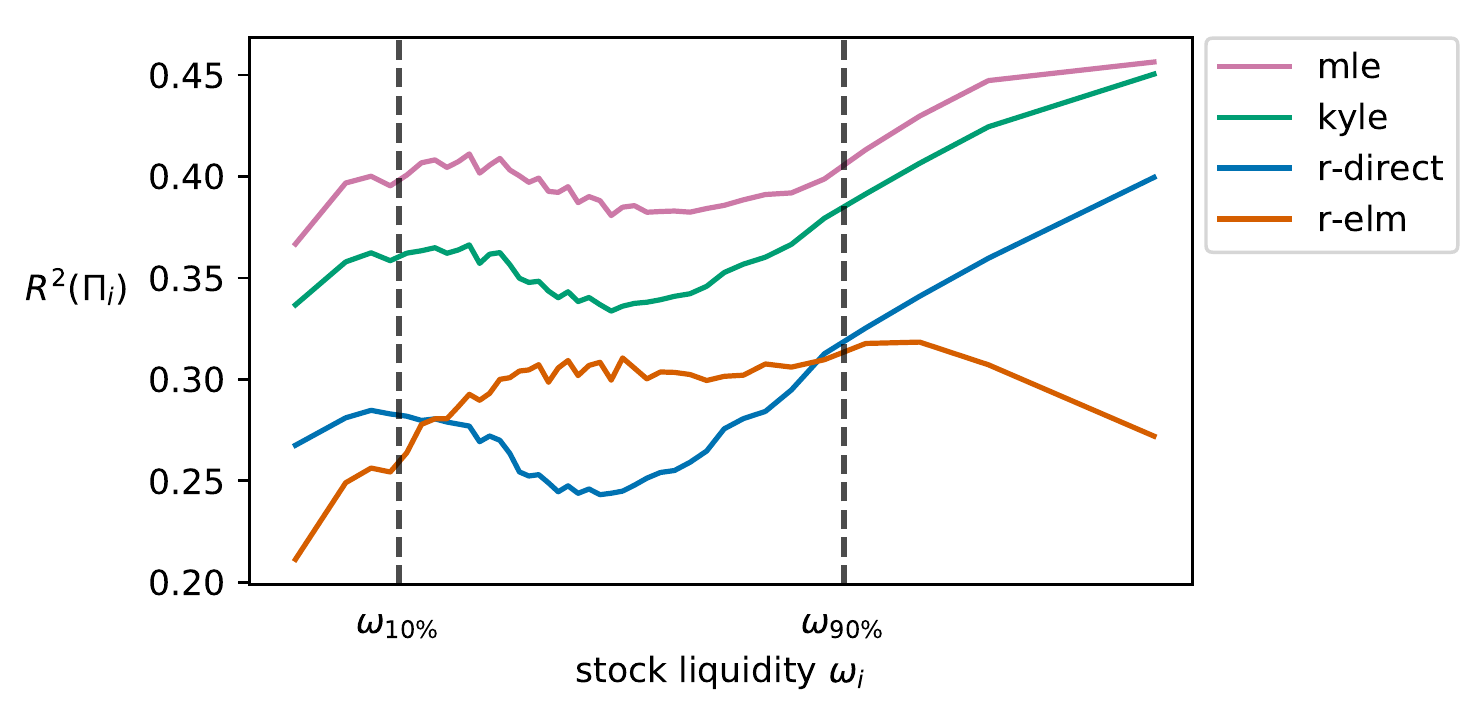}
		\caption{\textbf{Idiosyncratic scores as a function of liquidity.\\} For each stock in our dataset, we compute the the in-sample stock-specific scores $\rsq(\Pi_i)$ scores on 2016 data. We then represent the average in-sample stock-specific score as a function of the liquidity $\omega_i$, binning data by $\omega_i$ to smooth out noise. Results for the \rmle{} (in pink), \pkyle{} (in green), \rdir{} (in blue) and \relm{} (in orange) models are shown. We have further indicated the 10\% and 90\% quantiles of liquidity $\omega_{10\%}$ and $\omega_{90\%}$.}
	    \label{figure:r2_liquidity}
\end{figure}

An interesting feature of our stocks dataset is the heterogeneous liquidity among assets. This allows us to explore the influence of the liquidity of a given stock on the performance of different models. \cref{figure:r2_liquidity} shows the results of this analysis. Consistent with \cref{fig:gof}, we find that overall, in score terms, \rmle{}>\pkyle{}>\rdir{}>\relm{}. The \rdir{} model fares better for liquid stocks, where a larger fraction of variance can be explained by same-stock trades. Surprisingly, the same holds for \rmle{} and \pkyle{} models. The \relm{} model stands as an exception. It better explains price moves for stocks which are within the band of typical liquidity, between $\omega_{10\%}$ and $\omega_{90\%}$. This makes sense since the \relm{} model is self-stable as it aggregates liquidity of all stocks. Thus, though this assumption is justified for stocks of liquidity close to the average, it is violated outside of this zone. The \rmle{} and \pkyle{} models are not self-stable and better deal with very liquid or illiquid stocks. To further reinforce this point, for stocks of liquidity close to the average in our pool of stocks, the difference scores of the \pelm{} and \pkyle{} models reach a minimum. This is consistent with the fact that in the approximation $\Omega \approx \omega^2_{50\%} \id$, the two models coincide. Thus, violating self-stability (\cref{Axiom:self-stability}) is key to explain price changes for all ranges of liquidity within a basket of instruments.

\subsection{Robustness of goodness-of-fit} 

\begin{figure}[t!]
	\centering
		\centering
		\includegraphics[width=\linewidth]{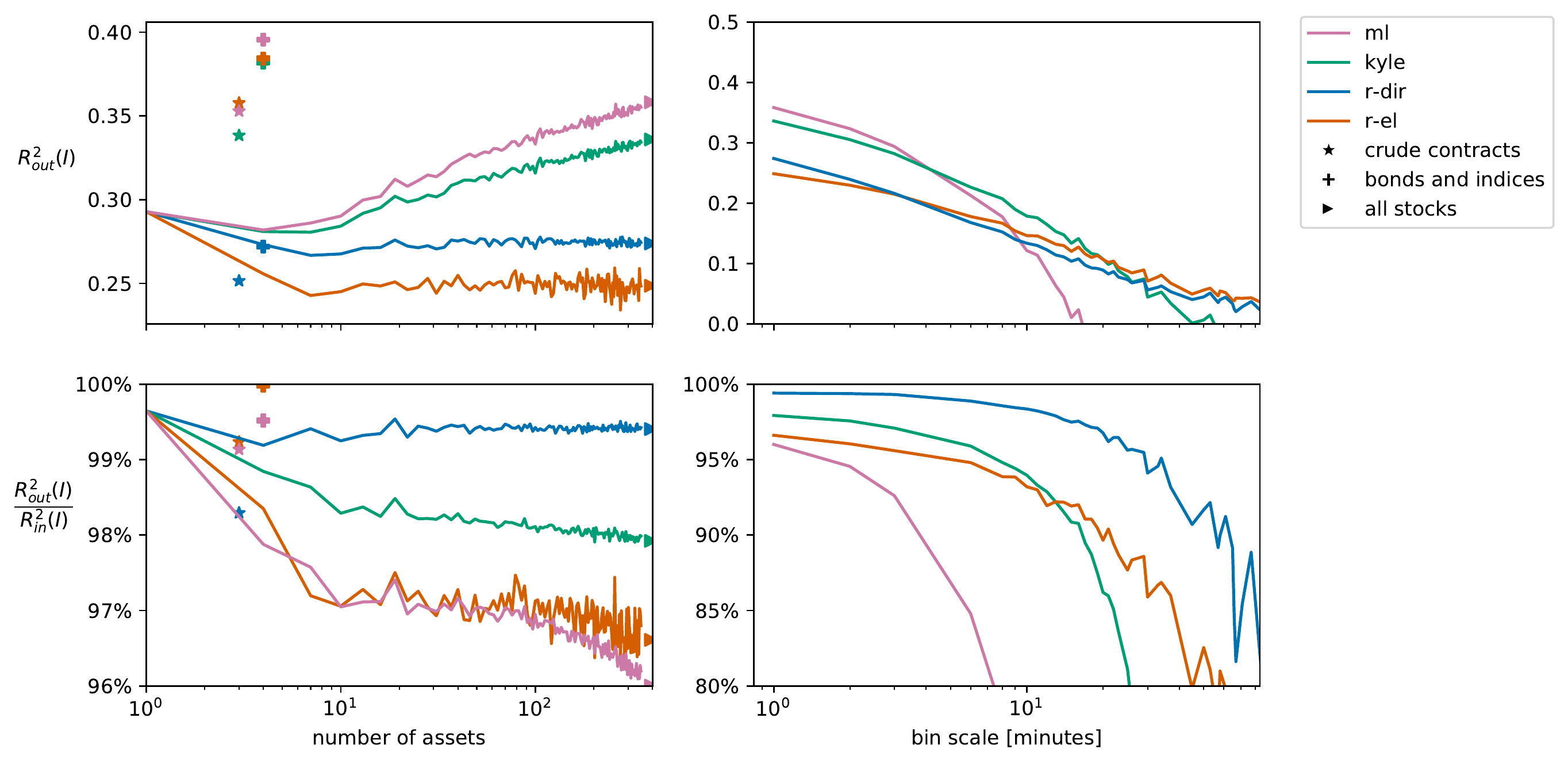}
		\caption{\textbf{Idiosyncratic score and overfitting as a function of the number of assets and bin timescale.} \\
		\textbf{Left column:} average out-of-sample idiosyncratic score $\rsq_{\textnormal{out}}(I_{\boldm{\sigma}})$ (top left) and overfitting coefficient $\frac{\rsq_{\textnormal{out}}(I_{\boldm{\sigma}})}{\rsq_{\textnormal{in}}(I_{\boldm{\sigma}})}$ (bottom left) computed using stocks data. Out-of-sample and in-sample scores were computed by randomly selecting a subset of stocks and computing scores on the given subset, repeating the procedure more when there are fewer stocks are selected than when a large proportion of stocks from our sample is considered. The average score for each models across all samples is then shown. Scores are shown for the \rmle{} (in pink), the \pkyle{} (in green), \rdir{} (in blue) and \relm{} (in orange). Stars show results for crude contracts, crosses for bonds and indices and triangles for all 393 stocks of our sample. \textbf{Right column:} idiosyncratic scores (top right) and overfitting coefficient (bottom right) as a function of the bin timescale. Scores were computed using the same procedure described in \cref{sec:stocks}, varying the bin parameter from 10 seconds up to around an hour.}
	    \label{figure:performance_nassets}
\end{figure}

The previous section compared the descriptive power of different cross-impact models. However, robustness of the different cross-impact models is also of interest. In \cref{figure:performance_nassets}, we show the out-of-sample score and overfitting coefficient for idiosyncratic price changes for our set of 393 stocks, as a function of the bin timescale and number of instruments. 
\\ \\
As expected, the number of degrees of freedom controls the overfitting of different models. This explains why, in terms of overfitting with respect to the number of instruments at the minute timescale, $\rdir{}<\pkyle{}<\rmle{} \approx \relm{}$. In contrast, models overfit less on futures, which suggests that overfitting decreases as the pairwise correlation between instruments increases. Furthermore, out-of-sample idiosyncratic scores for the \rmle{} and \pkyle{} model increase with the number of assets. A somewhat surprising result, despite the small pairwise correlation of instruments in our stock dataset and the large number of stocks considered in this study, is that idiosyncratic scores appear to keep increasing for more than $400$ assets. This suggests that there is still latent explanatory power in the dataset but only two models manage to extract it.
\\ \\
Focusing on the influence of the bin timescale, there is little overfitting at the minute timescale but it increases with the bin timescale. In particular, the good fit of the \rmle{} at small timescales quickly breaks down for larger timescales. On the other hand, both the \relm{} and \pkyle{} models are quite robust up until the 10 minute timescale. At this timescale, we expect our static approximation described in \cref{sec:setup_assumptions} to break down.

\section*{Conclusion}
\label{sec:conclusion}

Let us summarize what we have achieved.  Our main objective was to find suitable static cross-impact models given a set of empirical observations encoded in the sufficient statistics $(\Sigma, \Omega, R)$ which fit data well and led to well-behaved market dynamics. To do so, we introduced axioms, desirable properties of cross-impact models. We classified existing cross-impact models and characterised those which satisfy certain subset of axioms.
\\ \\
In all markets studied, our analysis confirms that cross-impact models are well suited to explain price changes, showing significant improvement compared to impact models in which cross-sectional effects are disregarded (see~\cref{fig:gof}). However, only the \pkyle{} and \rmle{} models perform well on all markets studied, whereas only the \pkyle{} model prevents arbitrage and is well-behaved when trading both liquid and illiquid instruments. This makes it an ideal model for other applications, such as optimal execution.
\\ \\
Independently of our specific model implementations, we also established a few characterisation results of cross-impact models from axioms. In particular, we showed that symmetry axioms (\cref{Axiom:split_invariance,Axiom:invariance_rotation,Axiom:permutational_invariance,Axiom:direct_invariance,Axiom:cash_invariance}) alone completely characterise return-based cross-impact models and that the \pkyle{} model is the only model to satisfy all cross-impact axioms.
\\ \\
Even though we have considered a linear, single-period scenario, the ideas introduced in this paper could be generalised. For instance, the framework can be adapted to deal with derivatives \cite{tomas2021cross}. Another topic is the generalisation of this framework to account for the auto-correlation of the order flow. This question is examined in \cite{Tomas2019FromModels, rosenbaum2021characterisation} but an axiom-first approach is still lacking. 
\\ \\
Finally, our results can be used to choose adequate cross-impact models for applications discussed in the literature. For optimal trading applications \cite{ekren2018optimal,ekren2019portfolio,garleanu2016dynamic,Lehalle2019ACorrelations}, it is natural to favor models which satisfy arbitrage axioms to prevent ill-behaved trading strategies. When modeling cross-impact at the microstructural level \cite{tsoukalas2019dynamic} then arbitrage, fragmentation and liquidity axioms are all important to rule out price manipulation. For each domain, we highlighted which cross-impact models would be good candidates.

\bibliographystyle{plain}
\bibliography{bibliography_main}

\clearpage

\appendix
\section{Proofs}
\label{app:proofs_models}

This section contains proofs of some results stated in the main text.

\subsection{Proof of \cref{prop:liq_to_frag_main,prop:fragm_to_liq_main}}
\label{sec:fragm_liq}

In this section, we establish some links between fragmentation and liquidity axioms. To do so, in the rest of this section, we will make use of two kinds of regularised covariance and response matrices. Given a linear subspace $V$, we first introduce the order flow regularised estimators:
\begin{align*}
    \Sigmaqeps &:= \Sigma \\
    \Omegaqeps &:= (\bar \Pi_V + \varepsilon \Pi_V) \Omega(\bar \Pi_V + \varepsilon \Pi_V) \\
    \Rqeps &:= R (\bar \Pi_V + \varepsilon \Pi_V) \, .
\end{align*}
These correspond to the multiplication of liquidity of instruments in $V$ by $\varepsilon$. Similarly, we introduce the price regularised estimators:
\begin{align*}
    \Sigmapeps &:= (\bar \Pi_V + \varepsilon \Pi_V) \Sigma (\bar \Pi_V + \varepsilon \Pi_V) \\
    \Omegapeps &:= \Omega \\
    \Rpeps &:= (\bar \Pi_V + \varepsilon \Pi_V) R \, .
\end{align*}
These correspond to the multiplication of price fluctuations of instruments in $V$ by $\varepsilon$. We begin with a convenient lemma that relates liquidity properties to fragmentation properties.

\begin{lemma}
    \label{lemma:expansion_fragmentation_liquidity}
    Let $\Lambda$ be a split-invariant and rotation-invariant (\cref{Axiom:invariance_rotation,Axiom:split_invariance}) cross-impact model and a linear subspace $V$ such that $\emptyset \subset V \subseteq \mathbb{R}^n$.
    Then, for all $\varepsilon > 0$, we have
    \begin{align*}
        \Lambda(\Sigmaqeps, \Omegaqeps, \Rqeps) 
        =& \bar \Pi_V \Lambda(\Sigmapeps, \Omegapeps, \Rpeps) \bar \Pi_V \\
        & +\varepsilon^{-1} \left[ \bar \Pi_V \Lambda(\Sigmapeps, \Omegapeps, \Rpeps) \Pi_V 
        +  \Pi_V \Lambda(\Sigmapeps, \Omegapeps, \Rpeps) \bar \Pi_V \right] \\
        & +  \varepsilon^{-2} \Pi_V \Lambda(\Sigmapeps, \Omegapeps, \Rpeps) \Pi_V \, .
    \end{align*}
\end{lemma}

\begin{proof}
Let $(u_1, \cdots, u_k)$ an orthonormal basis for the linear subspace $V$, where $k = \dim(V)$ and $(u_1, \cdots, u_\nassets)$ a completed orthonormal basis on $\R^{\nassets}$. We write $A := (u_1 \mid u_2 \mid \cdots \mid u_k) \in \mat{\nassets,k}$ and $U := (u_1 \mid u_2 \mid \cdots \mid u_\nassets) \in \mat{\nassets,\nassets}$. Then $\Pi_V = A A^\top$ and $\bar \Pi_V = I - A A^\top$. By rotation invariance, we have
$$
U \Lambda(\Sigmaqeps, \Omegaqeps, \Rqeps) U^\top = \Lambda(U \Sigmaqeps U^\top, U \Omegaqeps U^\top, U \Rqeps U^\top) \, .
$$
Since $(u_1, \cdots, u_\nassets)$ is an orthonormal basis, $UA$ only has non-zero entries along the diagonal. Writing $\widehat A := UA$ we can apply split invariance with $D = (I - \widehat{A} \widehat{A}^\top + \varepsilon \widehat{A}\widehat{A}^\top)$ to obtain
$$
U \Lambda(\Sigmaqeps, \Omegaqeps, \Rqeps) U^\top = D \Lambda(D^{-1} U \Sigmaqeps U^\top D^{-1}, D U \Omegaqeps U^\top D, D^{-1} U \Rqeps U^\top D) D \, .
$$
Straightforward computations show that
\begin{align*}
    D^{-1} U \Sigmaqeps U^\top D^{-1} = U \Sigmapeps U^\top \, \hspace{10pt} D U \Omegaqeps U^\top D = U \Omegapeps U^\top \, \hspace{10pt} D^{-1} U \Rqeps U^\top D = U \Rpeps U^\top \, .
\end{align*}
Therefore
\begin{align*}
U \Lambda(\Sigmaqeps, \Omegaqeps, \Rqeps) U^\top =& (I - \widehat{A} \widehat{A}^\top + \varepsilon^{-1} \widehat{A}\widehat{A}^\top) \Lambda \left(U \Sigmapeps U^\top, U \Omegapeps U^\top, U \Rpeps U^\top \right) (I - \widehat{A} \widehat{A}^\top + \varepsilon^{-1} \widehat{A}\widehat{A}^\top) \, .
\end{align*}
Applying rotational invariance once again we get
\begin{align*}
U \Lambda(\Sigmaqeps, \Omegaqeps, \Rqeps) U^\top =& U (I - A A^\top) \Lambda \left( \Sigmapeps,  \Omegapeps, \Rpeps \right) (I - A A^\top) U^\top \\
& + \varepsilon^{-1} U \left[ (I - A A^\top) \Lambda \left( \Sigmapeps,  \Omegapeps, \Rpeps \right) A A^\top + A A^\top \Lambda \left( \Sigmapeps,  \Omegapeps, \Rpeps \right) (I - A A^\top) \right] U^\top \\
& + \varepsilon^{-2} U A A^\top \Lambda \left( \Sigmapeps,  \Omegapeps, \Rpeps \right) A A^\top U^\top \, .
\end{align*}
This finally yields
\begin{align*}
\Lambda(\Sigmaqeps, \Omegaqeps, \Rqeps) =& \bar \Pi_V \Lambda \left( \Sigmapeps,  \Omegapeps, \Rpeps \right) \bar \Pi_V \\
& + \varepsilon^{-1} \left[ \bar \Pi_V \Lambda \left( \Sigmapeps,  \Omegapeps, \Rpeps \right) \Pi_V + \Pi_V \Lambda \left( \Sigmapeps,  \Omegapeps, \Rpeps \right) \bar \Pi_V \right] \\
& + \varepsilon^{-2} \Pi_V \Lambda \left( \Sigmapeps,  \Omegapeps, \Rpeps \right) \Pi_V \, ,
\end{align*}
which concludes the proof.
\end{proof}

In a similar fashion as \cref{lemma:expansion_fragmentation_liquidity}, one can prove the following Lemma.

\begin{lemma}
    \label{lemma:expansion_liquidity_fragmentation}
    Let $\Lambda$ be a split-invariant and rotation-invariant (\cref{Axiom:invariance_rotation,Axiom:split_invariance}) cross-impact model and a subspace $V$ such that $\emptyset \subset V \subseteq \mathbb{R}^n$. Then we have
    \begin{align*}
        \label{eq:lemma_fragm_split}
        \Lambda(\Sigmapeps, \Omegapeps, \Rpeps) 
        =& \bar \Pi_V \Lambda(\Sigmaqeps, \Omegaqeps, \Rqeps) \bar \Pi_V \\
        & +\varepsilon^{-1} \left[ \bar \Pi_V \Lambda(\Sigmaqeps, \Omegaqeps, \Rqeps) \Pi_V
        +  \Pi_V \Lambda(\Sigmaqeps, \Omegaqeps, \Rqeps) \bar \Pi_V \right] \\
        & + \varepsilon^{-2} \Pi_V \Lambda(\Sigmaqeps, \Omegaqeps, \Rqeps) \Pi_V.
    \end{align*}
\end{lemma}

\cref{lemma:expansion_liquidity_fragmentation,lemma:expansion_fragmentation_liquidity} enable us to relate cross-stability to fragmentation invariance. This is the topic of the next proposition. 

\begin{propositionp}{\ref*{prop:fragm_to_liq_main}}
\label{prop:fragm_to_liq}
Let $\Lambda$ be a jointly continuous cross-impact model which satisfies split and rotational invariance~(\cref{Axiom:invariance_rotation,Axiom:split_invariance}). Then
\begin{enumerate}[(i)]
    \item If $\Lambda$ satisfies semi-strong fragmentation invariance~(\cref{Axiom:semi_strong_fragmentation_invariance}), then it is weakly cross-stable (\cref{Axiom:weak_cross_stability}).
    \item If $\Lambda$ is strongly fragmentation invariant (\cref{Axiom:strong_fragmentation_invariance}), then it is strongly cross-stable (\cref{Axiom:strong_cross_stability}).
\end{enumerate}
\end{propositionp}

\begin{proof}
We first prove (i). Since the cross-impact model $\Lambda$ is continuous and satisfies semi-strong fragmentation invariance we have 
\begin{align*}
\Lambda(\Sigmapeps, \Omegapeps, \Rpeps) \Pi_V & \underset{\varepsilon \to 0}{\to} 0 \\
\Pi_V \Lambda(\Sigmapeps, \Omegapeps, \Rpeps) & \underset{\varepsilon \to 0}{\to} 0 \, .
\end{align*}
Plugging the above in the results of \cref{lemma:expansion_liquidity_fragmentation} yields
\begin{align*}
    \bar \Pi_V \Lambda(\Sigmapeps, \Omegapeps, \Rpeps) \Pi_V  = \varepsilon^{-1} \bar \Pi_V \Lambda(\Sigmaqeps, \Omegaqeps, \Rqeps) \Pi_V & \underset{\varepsilon \to 0}{\to} 0 \\
    \Pi_V \Lambda(\Sigmapeps, \Omegapeps, \Rpeps) \bar \Pi_V  = \varepsilon^{-1} \Pi_V \Lambda(\Sigmaqeps, \Omegaqeps, \Rqeps) \bar \Pi_V & \underset{\varepsilon \to 0}{\to} 0 \, .
\end{align*}
Thus $\Lambda$ is weakly cross-stable. We now prove (ii). Continuity at $\varepsilon=0$ and strong fragmentation invariance yield
$$
    \bar \Pi_V \Lambda(\Sigmapeps, \Omegapeps, \Rqeps) \bar \Pi_V \underset{\varepsilon \to 0}{\to} \bar \Pi_V \Lambda(\bar \Pi_V \Sigma \bar \Pi_V, \bar \Pi_V \Omega \bar \Pi_V, \bar \Pi_V R \bar \Pi_V) \bar \Pi_V  \, .
$$
Plugging the above into the results of \cref{lemma:expansion_liquidity_fragmentation} gives
$$
        \bar \Pi_V \Lambda(\Sigmaqeps, \Omegaqeps, \Rqeps) \bar \Pi_V
        = \bar \Pi_V \Lambda(\Sigmapeps, \Omegapeps, \Rpeps) \bar \Pi_V \underset{\varepsilon \to 0}{\to} \bar \Pi_V \Lambda(\bar \Pi_V \Sigma \bar \Pi_V, \bar \Pi_V \Omega \bar \Pi_V, \bar \Pi_V R \bar \Pi_V) \bar \Pi_V \, .
$$
This implies that $\Lambda$ is strongly cross-stable.
\end{proof}

Interestingly, the converse of \cref{prop:fragm_to_liq} does not hold, thus indicating that the fragmentation invariance properties play a more fundamental role than liquidity related axioms. The next proposition shows the converse, provided some additional regularity of the cross-impact model.

\begin{propositionp}{\ref*{prop:liq_to_frag_main}}
\label{prop:liq_to_frag}
Let $\Lambda$ be a jointly continuous cross-impact model which satisfies split and rotational invariance~(\cref{Axiom:invariance_rotation,Axiom:split_invariance}). We further assume that, for every linear subspace $V$ and using the previous notations, $\varepsilon^{2} \Lambda(\Sigmaqeps, \Omegaqeps, \Rqeps)  \underset{\varepsilon \to 0}{\to} 0$. Then
\begin{enumerate}[(i)]
    \item If $\Lambda$ is weakly cross-stable (\cref{Axiom:weak_cross_stability}), then it satisfies semi-strong fragmentation invariance~(\cref{Axiom:semi_strong_fragmentation_invariance}).
    \item If $\Lambda$ is strongly cross-stable (\cref{Axiom:strong_cross_stability}), then it is strongly fragmentation invariant (\cref{Axiom:strong_fragmentation_invariance}). 
\end{enumerate}
\end{propositionp}

\begin{proof}
We first prove (i). Since the cross-impact model $\Lambda$ is weakly cross-stable we have 
\begin{align*}
\bar \Pi_V \Lambda(\Sigmaqeps, \Omegaqeps, \Rqeps) \Pi_V & \underset{\varepsilon \to 0}{\to} 0 \\
\Pi_V \Lambda(\Sigmaqeps, \Omegaqeps, \Rqeps) \bar \Pi_V & \underset{\varepsilon \to 0}{\to} 0 \, .
\end{align*}
Furthermore, by assumption we have 
$$
\varepsilon^{2} \Lambda(\Sigmaqeps, \Omegaqeps, \Rqeps)  \underset{\varepsilon \to 0}{\to} 0.
$$
Plugging the above in the results of \cref{lemma:expansion_fragmentation_liquidity} yields
\begin{align*}
    \bar \Pi_V \Lambda(\Sigmaqeps, \Omegaqeps, \Rqeps) \Pi_V &= \varepsilon^{-1} \bar \Pi_V \Lambda(\Sigmapeps, \Omegapeps, \Rpeps) \Pi_V   \underset{\varepsilon \to 0}{\to} 0 \\
    \Pi_V \Lambda(\Sigmaqeps, \Omegaqeps, \Rqeps) \bar \Pi_V &= \varepsilon^{-1} \Pi_V \Lambda(\Sigmapeps, \Omegapeps, \Rpeps) \bar \Pi_V   \underset{\varepsilon \to 0}{\to} 0 \\
    \varepsilon^2 \Pi_V \Lambda(\Sigmaqeps, \Omegaqeps, \Rqeps) \Pi_V &= \Pi_V \Lambda(\Sigmapeps, \Omegapeps, \Rpeps) \bar \Pi_V \underset{\varepsilon \to 0}{\to} 0 \, .
\end{align*}
Combining the above and using continuity, we obtain 
\begin{align*}
\Pi_V \Lambda(\Sigmapeps, \Omegapeps, \Rpeps) & \underset{\varepsilon \to 0}{\to} 0 = \Pi_V \Lambda(\bar \Pi_V \Sigma \bar \Pi_V , \Omega, \bar \Pi_V R) \\
\Lambda(\Sigmapeps, \Omegapeps, \Rpeps) \Pi_V & \underset{\varepsilon \to 0}{\to} 0 = \Lambda(\bar \Pi_V \Sigma \bar \Pi_V , \Omega, \bar \Pi_V R) \Pi_V \, .
\end{align*}
Thus this proves that $\Lambda$ is semi-strongly fragmentation invariant. We now prove (ii). Continuity at $\varepsilon=0$ and strong cross-stability yield
$$
    \bar \Pi_V \Lambda(\Sigmaqeps, \Omegaqeps, \Rqeps) \bar \Pi_V \underset{\varepsilon \to 0}{\to} \bar \Pi_V \Lambda(\bar \Pi_V \Sigma \bar \Pi_V, \bar \Pi_V \Omega \bar \Pi_V, \bar \Pi_V R \bar \Pi_V) \bar \Pi_V  \, .
$$
Plugging the above into the results of \cref{lemma:expansion_fragmentation_liquidity} gives
$$
        \bar \Pi_V \Lambda(\Sigmaqeps, \Omegaqeps, \Rqeps) \bar \Pi_V
        = \bar \Pi_V \Lambda(\Sigmapeps, \Omegapeps, \Rpeps) \bar \Pi_V \underset{\varepsilon \to 0}{\to} \bar \Pi_V \Lambda(\bar \Pi_V \Sigma \bar \Pi_V, \bar \Pi_V \Omega \bar \Pi_V, \bar \Pi_V R \bar \Pi_V) \bar \Pi_V \, .
$$
This implies that $\Lambda$ is strongly fragmentation invariant.
\end{proof}

\cref{prop:liq_to_frag,prop:fragm_to_liq} show that fragmentation and cross-stability axioms are related. Furthermore, for cross-impact models which satisfy the regularity property of \cref{prop:liq_to_frag}, the two sets of axioms are equivalent.

\subsection{Proof of \cref{prop:kyle_sym_unique_main,prop:kyle_arb_cov_unique_main}}
\label{sec:connections}
In this section, we characterise the models which satisfy the axioms introduced in \cref{sec:axioms}. We begin with the following proposition, the proof of which is heavily inspired by~\cite{Caballe1994ImperfectNeutrality,delMolino2018TheDifferent}.

\begin{propositionp}{\ref*{prop:kyle_arb_cov_unique_main}}
\label{prop:kyle_arb_cov_unique}
Let $\Lambda$ be a symmetric, positive-semidefinite and return covariance consistent cross-impact model (\cref{Axiom:symmetric,Axiom:positive_semidefinite,Axiom:consistency_correlation}).
Then $\Lambda = \Lambda_{\emph{\pkyle}}$ up to a multiplicative constant.
\end{propositionp}

\begin{proof}
	Let $\Lambda$ be a cross-impact model which satisfies \cref{Axiom:positive_semidefinite,Axiom:consistency_correlation} and $(\Sigma, \Omega, R) \in (\spd{n} \times \pd{n} \times \mat{n})$. We assume for convenience that the multiplicative constant in \cref{Axiom:consistency_correlation} is one. Writing $\Lambda$ for $\Lambda(\Sigma, \Omega, R) $, and $\Lc$ for a matrix such that $\Omega = \Lc \Lc^\top$, we have
	\begin{equation*}
	\Sigma = \Lambda \Omega \Lambda^\top  = \Lambda \Lc \Lc^\top \Lambda^\top = (\Lambda \Lc) (\Lambda \Lc)^\top.
	\end{equation*}
	Thus, by unicity up to a rotation of the square root decomposition, writing $\Gc$ for a matrix such that $\Sigma = \Gc \Gc^\top$, there exists a rotation $O$ such that $\Lambda  =  \Gc O \Lc^{-1}$. Furthermore, since $\Lambda$ is symmetric,
	$$
	\Gc O \Lc^{-1} = (\Gc O \Lc^{-1})^\top.
	$$
	Rewriting, we find
	$$
	\Lc^\top \Gc O = O^\top \Gc^\top \Lc,
	$$
	so that the matrix $\Lc^\top \Gc O$ is symmetric and satisfies
	\begin{align*}
	(\Lc^\top \Gc O) (\Lc^\top \Gc O)^\top = (\Lc \Gc) (\Lc^\top \Gc)^\top \, .
	\end{align*}
	Since $(\Lc^{\top} \Gc) (\Lc^{\top} \Gc)^\top$ is symmetric positive semi-definite, the symmetric square root is unique and
	$$
	\Lc^\top \Gc O = \sqrt{(\Lc^{\top} \Gc) (\Lc^{\top} \Gc)^\top}.
	$$
	Plugging this back into the expression of the cross-impact matrix yields the result:
	$$
	\Lambda = \Gc O \Lc^{-1} = \Lc^{-\top} \sqrt{(\Lc^{\top} \Gc) (\Lc^{\top} \Gc)^\top} \Lc^{-1} = \Lc^{-\top} \sqrt{\Lc^{\top} \Sigma \Lc} \Lc^{-1} \, .
	$$
\end{proof}

Hence, there is a single symmetric, positive-semidefinite, covariance-consistent, cross-impact model. Given that the fragmentation-related axioms seem so fundamental, one might wonder how many models one can build that satisfy that family of properties. Surprisingly, we find that the class of models enjoying both split invariance and rotational invariance is quite small, as shown in the next lemma.

\begin{lemma}
    \label{lemma:change_of_basis_diagonal_omega}
	Let $\Lambda$ be a cross-impact model which satisfies \cref{Axiom:invariance_rotation,Axiom:split_invariance}. Then, for all $(\Sigma, \Omega, R) \in (\spd{n} \times \pd{n} \times \mat{n})$, it can be written as
	$$
	\label{eq:lemma_rot_split}
	\Lambda(\Sigma, \Omega, R) = \Lc^{-\top} U \Lambda(
	U^\top \hat \Sigma U,
	\id,
	U^\top \hat R U) U^\top \Lc^{-1},
	$$
	where
	\begin{align*}
	    \Omega &= \Lc \Lc^\top \\
	    \hat \Sigma &= \Lc^{\top} \Sigma \Lc \\
	    \hat R &= \Lc^{\top} R \Lc^{-\top}
	\end{align*}
	and $U$ is an orthogonal matrix (i.e., $U U^\top = \id$).
\end{lemma}

\begin{proof}
The lemma is obtained by applying sequentially rotational invariance, split invariance and again rotational invariance. The first two transformations can be used in order to remove the dependency in $\Omega$ as the second argument of the $\Lambda(\Sigma, \Omega, R)$ function.
\end{proof}
When one discards the influence of the response matrix, the model can further be characterised as shown by the next proposition.

\begin{propositionp}{\ref*{prop:kyle_sym_unique_main}}
\label{prop:explicit_invariant_lambda}
A return covariance based cross-impact model $\Lambda$ that is both split-invariant and rotation-invariant (\cref{Axiom:invariance_rotation,Axiom:split_invariance}) can always be written in the form
	$$
	\Lambda(\Sigma, \Omega) = \Lc^{-\top} U F(\bmu) U^\top \Lc^{-1},
	$$
	where
		\begin{align*}
	    \Omega = \Lc \Lc^\top \quad ; \quad 
	    \hat \Sigma := \Lc^{\top} \Sigma \Lc \quad ; \quad 
	    U^\top \hat{\Sigma} U := \diag{\bmu} \quad ; \quad 
	    F(\bmu) := \Lambda(\diag{\bmu}, \id).
	\end{align*}
Furthermore, if $\Lambda$ is cash-invariant and direct-invariant \cref{Axiom:cash_invariance,Axiom:direct_invariance}, then $F(\boldm{\mu}) \propto \diag{\boldm{\mu}}^{1/2}$ and $\Lambda = \Lambda_{\emph{\pkyle}}$ up to a multiplicative constant.
\end{propositionp}

\begin{proof}
For a return covariance based model, we can simply choose from \cref{eq:lemma_rot_split} to fix $U$ as the rotation that diagonalizes the symmetric matrix $\hat \Sigma$, obtaining:
$$
U^\top \hat \Sigma U = \diag{\bmu} \, .
$$
This choice implies
\begin{align*}
    \Lambda(\Sigma, \Omega) 
    &= \Lc^{-\top} U \Lambda(
	\diag{\bmu},
	\id) U^\top \Lc^{-1},
\end{align*}
which yields the result of the first part of the proposition. Furthermore, if we assume $\Lambda$ is cash-invariant and direct-invariant (\cref{Axiom:cash_invariance,Axiom:direct_invariance}), we have
$$
\Lambda(
	\diag{\bmu},
	\id) = \sum_{i=1}^{d} \sqrt{\mu_i} \Lambda(\boldm{e_i} \boldm{e_i}^\top , \boldm{e_i} \, , \boldm{e_i}^\top)
$$
which yields the \pkyle{} model up to a constant.
\end{proof}

The above shows that the only return-based cross-impact model which satisfies all symmetry axioms \cref{Axiom:cash_invariance,Axiom:direct_invariance,Axiom:invariance_rotation,Axiom:split_invariance,Axiom:direct_invariance,Axiom:permutational_invariance} is the \pkyle{} model.

\subsection{Proof of important properties of the \pkyle{} model}
\label{proof:kyle}

This section is dedicated to showing that the \pkyle{} model satisfies all the axioms outlined in section \cref{sec:axioms}. As the fragmentation and invariance axioms were discussed in the previous section, the next lemma shows that the \pkyle{} model is also cross-stable.

\begin{lemma}
	The \pkyle{} model is strongly cross-stable in the sense of \cref{Axiom:strong_cross_stability,Axiom:self-stability} and is not self-stable in the sense of \cref{Axiom:self-stability}. 
\end{lemma}

\begin{proof}
	Let $V$ be a linear subspace of $\mathbb{R}^n$) and $\varepsilon>0$.  Note that, writing $\Gc$ for a matrix such that $\Gc \Gc^\top = \Sigma$, for any matrix $\Lc_{\varepsilon}$ such that $\Lc_{\varepsilon} \Lc_{\varepsilon}^\top = \Omegaqeps$, we previously showed that there exists a rotation matrix $O_{\varepsilon} = (\Lc^\top_{\varepsilon} \Gc)^{-1} \sqrt{(\Lc^{\top}_{\varepsilon} \Gc) (\Lc^{\top}_{\varepsilon} \Gc)^\top}$ such that we have
	$$
	\Lambda_{\pkyle}(\Sigmaqeps, \Omegaqeps, \Rqeps) = \Gc O_{\varepsilon} \Lc_{\varepsilon}^{-1}.
	$$
	However, $\Omegaqeps = (\bar \Pi_V + \varepsilon \Pi_V) \Omega (\bar \Pi_V + \varepsilon \Pi_V) = (\bar \Pi_V + \varepsilon \Pi_V) \Lc  \Lc^\top  (\bar \Pi_V + \varepsilon \Pi_V) = [(\bar \Pi_V + \varepsilon \Pi_V) \Lc] [(\bar \Pi_V + \varepsilon \Pi_V) \Lc]^\top$. Thus,
	\begin{align*}
	\Lambda_{\pkyle}(\Sigmaqeps, \Omegaqeps, \Rqeps) &= \Gc O_{\varepsilon} [(\bar \Pi_V + \varepsilon \Pi_V) \Lc]^{-1} \\
	&= \Gc O_{\varepsilon}\Lc^{-1} (\bar \Pi_V + \dfrac{1}{\varepsilon} \Pi_V) \\
	&= \Gc O_{\varepsilon}\Lc^{-1} \bar \Pi_V + \dfrac{1}{\varepsilon} \Gc O_{\varepsilon} \Lc^{-1} \Pi_V.
	\end{align*}
	Using the symmetry of the \pkyle{} model, the above yields:
	$$
	\Lambda_{\pkyle}(\Sigmaqeps, \Omegaqeps, \Rqeps) = \bar \Pi_V \Lc^{-\top}  O_{\varepsilon}^\top \Gc^\top +\dfrac{1}{\varepsilon} \Pi_V \Lc^{-\top}  O_{\varepsilon}^\top \Gc^\top.
	$$
	Thus, we have
	\begin{align*}
	\bar \Pi_V  \Lambda_{\pkyle}(\Sigmaqeps, \Omegaqeps, \Rqeps) \Pi_V &= \bar \Pi_V \Lc^{-\top}  O_{\varepsilon}^\top \Gc^\top \Pi_V \\
	\Pi_V  \Lambda_{\pkyle}(\Sigmaqeps, \Omegaqeps, \Rqeps) \bar \Pi_V &= \Pi_V \Gc O_{\varepsilon}\Lc^{-1} \bar \Pi_V \, .
	\end{align*}
	Since $O_{\varepsilon}^\top$ is an orthogonal matrix, we have
	\begin{align*}
	\bar \Pi_V  \Lambda_{\pkyle}(\Sigmaqeps, \Omegaqeps, \Rqeps) \Pi_V &\underset{\varepsilon \to 0}{=} O(1) \\
	\Pi_V  \Lambda_{\pkyle}(\Sigmaqeps, \Omegaqeps, \Rqeps) \bar \Pi_V &\underset{\varepsilon \to 0}{=} O(1) \, ,
	\end{align*}
	which proves that \pkyle{} satisfies \cref{Axiom:weak_cross_stability}. Furthermore, 
	$$
	\Pi_V  \Lambda_{\pkyle} \Pi_V = \dfrac{1}{\varepsilon} \Pi_V \Lc^{-\top}  O_{\varepsilon}^\top \Gc^\top \Pi_V,
	$$ 
	so that unless $\Pi_V \Lc^{-\top}  O_{\varepsilon}^\top \Gc^\top \Pi_V = 0$, we have:
	$$
       \lvert \lvert  \Pi_V  \Lambda_{\pkyle} \Pi_V \rvert \rvert = \varepsilon^{-1} \lvert \lvert \Pi_V
        \Lc^{-\top}  O_{\varepsilon}^\top 
	\Gc^\top \Pi_V \rvert \rvert \underset{\varepsilon \to 0}{\to} \infty \, .
	$$ 
	Choosing diagonal $\Sigma$ and $\Omega$ such that $\Pi_V \Lc \neq 0$ and $\Gc\Pi_V \neq 0$, we see that $\Pi_V \Lc^{-\top}  O_{\varepsilon}^\top \Gc^\top \Pi_V = 0$ cannot hold for all $\Sigma, \Omega$. This shows that \pkyle{} does not satisfy \cref{Axiom:self-stability}. 
	Finally, notice that by using \cref{lemma:change_of_basis_diagonal_omega} one can make $\Omega$ appear only in the combination $\Lc^\top \Sigma \Lc$, which is insensitive to the components of $\Omega$ belonging to the kernel of $\Sigma$, which proves strong cross-stability (\cref{Axiom:strong_cross_stability}).
\end{proof}

\clearpage

\section{Data}
\label{app:data}

This appendix contains details on the datasets and processing used to apply the different models.

\subsection{Crude contracts}
\label{app:data:crude}

\begin{figure}[t!]
    \centering
    \includegraphics[width=0.95\columnwidth]{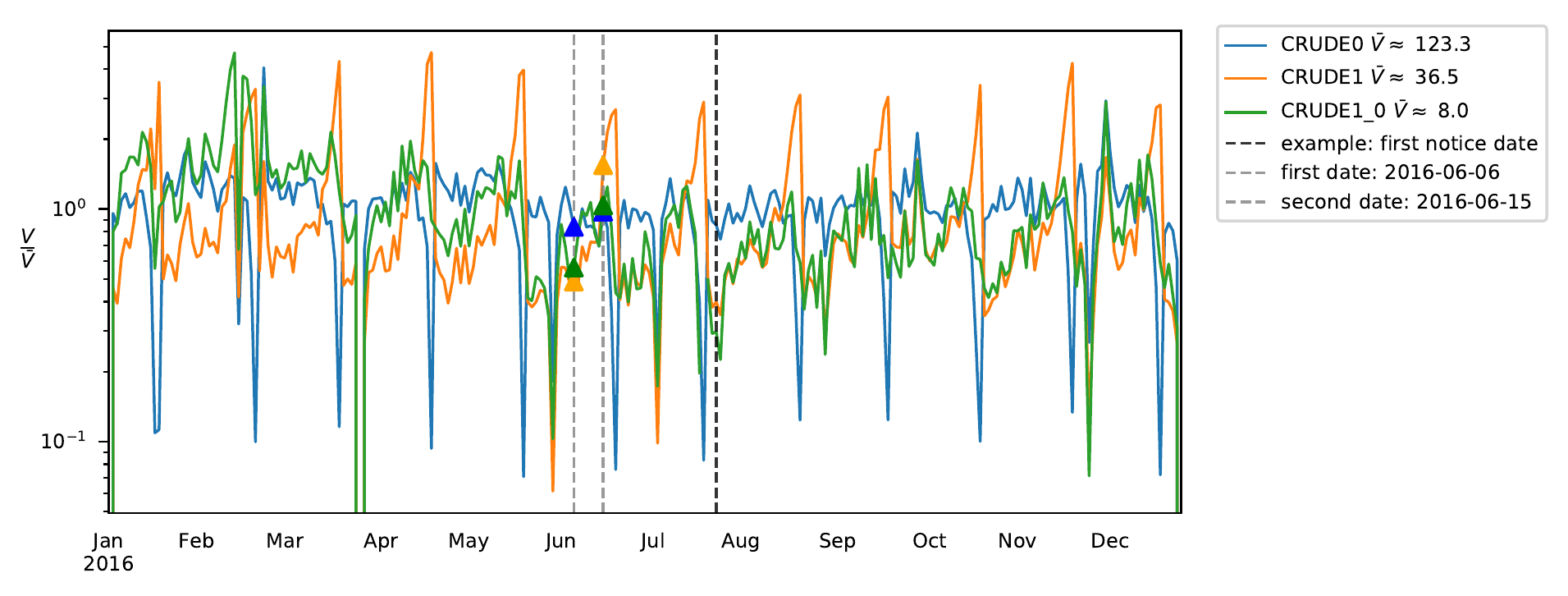}
    \caption{\textbf{Number of traded NYMEX Crude oil futures and Calendar Spread contracts (in thousands) relative to daily number of traded contracts.} \\ 
    The number of contracts sold relative to the daily average is shown for the front month contract (in blue), the subsequent month (in orange) and the Calendar Spread (in green). The average number of traded NYMEX Crude oil futures and Calendar Spread contracts $\bar{V}$ over 2016 is shown in the upper right corner. Vertical dashed lines show  specific dates. An example of first notice date for the front month contract is shown in bold black. After the first notice date, holders of the future contract may ask for physical delivery of the underlying. We also show two dates away from a first notice date: the 6th and 15th of June 2016. Colored triangles show the relative number of contracts exchanged on these dates. Note that the number of contracts is represented in thousands and was not adjusted by the basis point, so that the underlying of each contract is 1000 barrels of oil.}
    \label{fig:crude_heterogeneity_traded_volume}
\end{figure}

\paragraph*{Description of the dataset} We collected trades and quotes data from January 2016 to December 2017, between 9:30AM to 7:30PM UTC, where most of the trading takes place in our dataset, removing 30 minutes around the opening of trading hours to mitigate intraday seasonality. After filtering and processing, we have a total of 430 days in our sample (237 in 2016 and 193 in 2017). We highlight below two important features of our pre-processing for the estimation of $\Sigma$, $\Omega$ and $R$. 

\paragraph*{Pre-processing: accounting for non-stationarity} Overall, the front month contract CRUDE0 is by far the most liquid, followed by the subsequent month contract CRUDE1 and the calendar spread CRUDE1\_0. However, there are strong seasonal dependencies which are shown in \cref{fig:crude_heterogeneity_traded_volume}. For example, the subsequent month contract becomes more liquid as one approaches the maturity of the front month contract. Global estimators of $\Sigma$, $\Omega$ and $R$ would thus be biased by this varying liquidity $\boldm{\omega}$ ($\boldm{\sigma}$ also appears to follow a non-stationary pattern, but is not shown here). Thus, we used local (daily) estimators of price volatility $\boldm{\sigma_t}$ and liquidity $\boldm{\omega_t}$, and built local covariance estimators $\Sigma_t$ and $\Omega_t$ by assuming stationarity of the correlations $\varrho = \diag{\boldm{\sigma_t}}^{-1} \Sigma_t  \diag{\boldm{\sigma_t}}^{-1}$ and $\varrho_\Omega = \diag{\boldm{\omega_t}}^{-1} \Omega_t \diag{\boldm{\omega_t}}^{-1}$. We estimate volatility and liquidity with a simple standard deviation: $\sigma_{i,t}^2 = \langle \Delta p_{i,t}^2 \rangle$ and $\omega_{i,t}^2 = \langle q_{i,t}^2 \rangle$, where the average $\langle \cdot \rangle$ is computed using data on day $t$.

\begin{figure}[h!]
  \centering
    \includegraphics[width=0.8\linewidth]{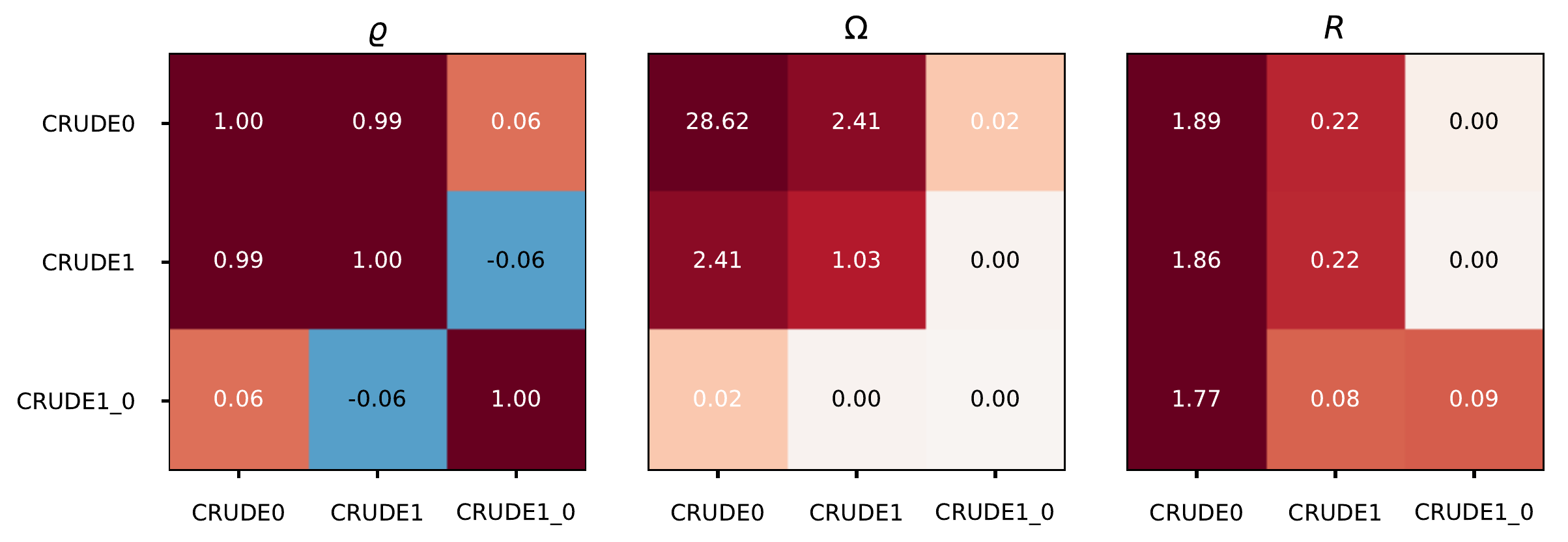}
    \caption{  \textbf{Estimates of $\varrho$, $\Omega$ and $R$ for Crude contracts (in MUSD). \\} The return correlation matrix $\rho$ (left), order flow covariance matrix $\Omega$ (center) and response matrix $R$ (right) were estimated using 2016 data and computed on the 6th of June 2016. This date represents the typical behavior of these contracts far away from the first notice date, before rolling effects become relevant. To highlight the amount of notional traded, order flow is reported in millions of exchanged dollars according to the average value of each contract on the 6th of June 2016. Though non-nill, order flow covariance of Calendar Spread thus appears small because traded notional is much smaller than on each leg of the futures contract.}
  \label{fig:crude_price_corr}
\end{figure}

\paragraph*{Structure of $\rho$, $\Omega$ and $R$} \cref{fig:crude_price_corr} reports the estimators of $\varrho$, $\Omega$ and $R$ for the 6th of June 2016. The figure shows $\varrho$ has one zero-volatility direction and one direction of very small fluctuations. Thus models which satisfy fragmentation invariance (\cref{Axiom:strong_fragmentation_invariance,Axiom:semi_strong_fragmentation_invariance,Axiom:weak_fragmentation_invariance}) should give better predictions. On the other hand, $\Omega$ highlights the difference in liquidity of our assets. Thus, we should be cautious of models which do not satisfy stability axioms. Indeed, these will not penalize trading directions of small liquidity. 

\paragraph*{Pre-processing: cleaning estimators}  As illustrated in \cref{fig:crude_price_corr}, where the structure of $\Sigma$, $\Omega$ and $R$ are shown for a typical day, one can appreciate that the correlation between the two future contracts CRUDE0 and CRUDE1 is close to one, whereas the correlation with the Calendar Spread contract is very small, due to the small volatility of the fluctuations along the relative mode. Because of these effects, the sign of the Calendar Spread correlations with CRUDE0 and CRUDE1 is non-trivial to estimate: due to microstructural effects, the measured correlation is dominated by tick-size related effects \footnote{To test this hypothesis, we estimated the empirical smallest eigenvalue of the covariance matrix for multiple futures contract as a function of relative tick size (not shown). If price changes of the Calendar Spread were given by the legs of the contract, this eigenvalue should be equal to zero. However, we found that as the tick size increases, so does the smallest eigenvalue away from zero. This thus validates our hypothesis and justifies the need for additional processing of futures data.}. In fact, empirical price changes of the Calendar Spread are not given by the difference of price changes of the legs. To solve this issue, we impose the price changes of the Calendar Spread according to the price changes of the futures contracts.

\begin{figure}[h!]
    \centering
    \includegraphics[width=0.8\linewidth]{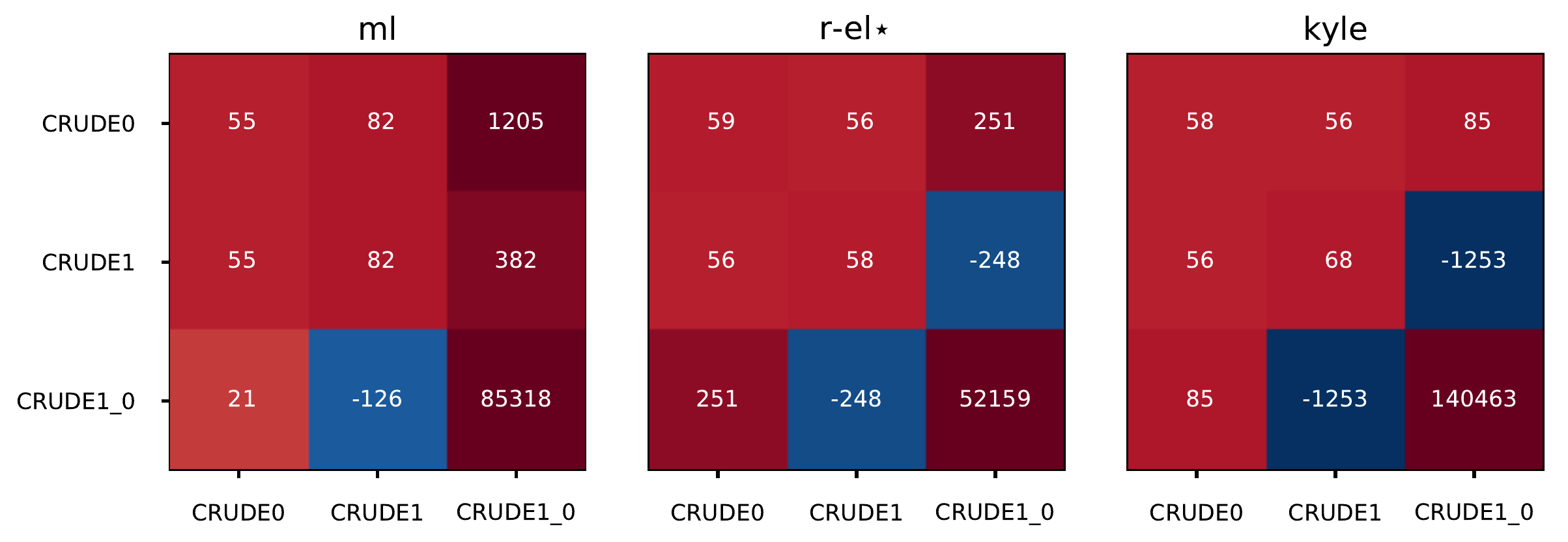}
    \caption{  \textbf{Values of different cross-impact models for Crude contracts.}\\
    We report the values of the \rmle{} (left), \relm{}$\star$ (center) and \pkyle{} (right) cross-impact models for the covariances of the 6th of June 2016. Units are chosen to represent the relative price change in basis points ($10^{-4}$ of the asset price) by hundred  million USD worth of contract traded.}
    \label{fig:impact_crude}
\end{figure}

\paragraph*{Cross-impact models for Crude oil contracts} \cref{fig:impact_crude} shows the calibrated \rmle{}, \relm{}$\star$ and \pkyle{} models. Each satisfies weak fragmentation invariance (\cref{Axiom:weak_fragmentation_invariance}). Therefore, they prevent arbitrage which would trade the physical Calendar Spread contract against the synthetic Calendar Spread (made up of CRUDE0 and CRUDE1). However, \rmle{} and \pkyle{} are not self-stable (\cref{Axiom:self-stability}) while the \relm{}$\star$ model is. This explains why impact from trading the illiquid Calendar Spread is much larger in the \rmle{} and \pkyle{} models than in the \relm{}$\star$ model.

\subsection{Bonds and indices}
\label{app:data:bonds}

\paragraph*{Description of the dataset} We look at 10-year US Treasury note futures and the E-MINI futures. We collect data from the Chicago Mercantile Exchange and use the first two upcoming maturities of both contracts (respectively called SPMINI and SPMINI3 for E-MINI contracts and 10USNOTE and 10USNOTE3 for 10-year US treasury notes). E-Mini futures are quarterly, financially settled contracts with maturities in March, June, September and December. At expiry, the final settlement price of E-MINI futures is a proxy for the S\&P500 index using the opening prices of the underlying stocks belonging to the index. Similarly, the 10-year treasury note futures are quaterly, financially settled contracts with maturities in March, June, September and December. At expiry, the final settlement price is volume weighted average price of past trades on the underlying treasury note.\footnote{This is a simplification of the settlement rules to emphasize the expected value of the final settlement price. Further details about the final settlement price of E-MINI futures and 10-year US Treasury Note futures can be found in the CME Rulebook.} We collected trades and quotes data from January 2016 to December 2017, between 9AM to 7PM UTC, where most of the trading takes place in our dataset. After filtering days for which data for one product was missing, we keep a total of 160 days (75 in 2016 and 85 in 2017). We highlight below one important pre-processing step for the estimation of $\Sigma$, $\Omega$ and $R$.

\paragraph*{Pre-processing: accounting for non-stationarity} The same non-stationary behavior observed for Crude Oil futures contract is observed here. Thus we adopt the same estimation procedure for the local covariance estimators $\Sigma_t$ and $\Omega_t$ by assuming stationarity of the correlations $\varrho = \diag{\boldm{\sigma_t}}^{-1} \Sigma_t  \diag{\boldm{\sigma_t}}^{-1}$ and $\varrho_\Omega = \diag{\boldm{\omega_t}}^{-1} \Omega_t \diag{\boldm{\omega_t}}^{-1}$.

\begin{figure}[H]
    \centering
    \includegraphics[width=0.8\linewidth]{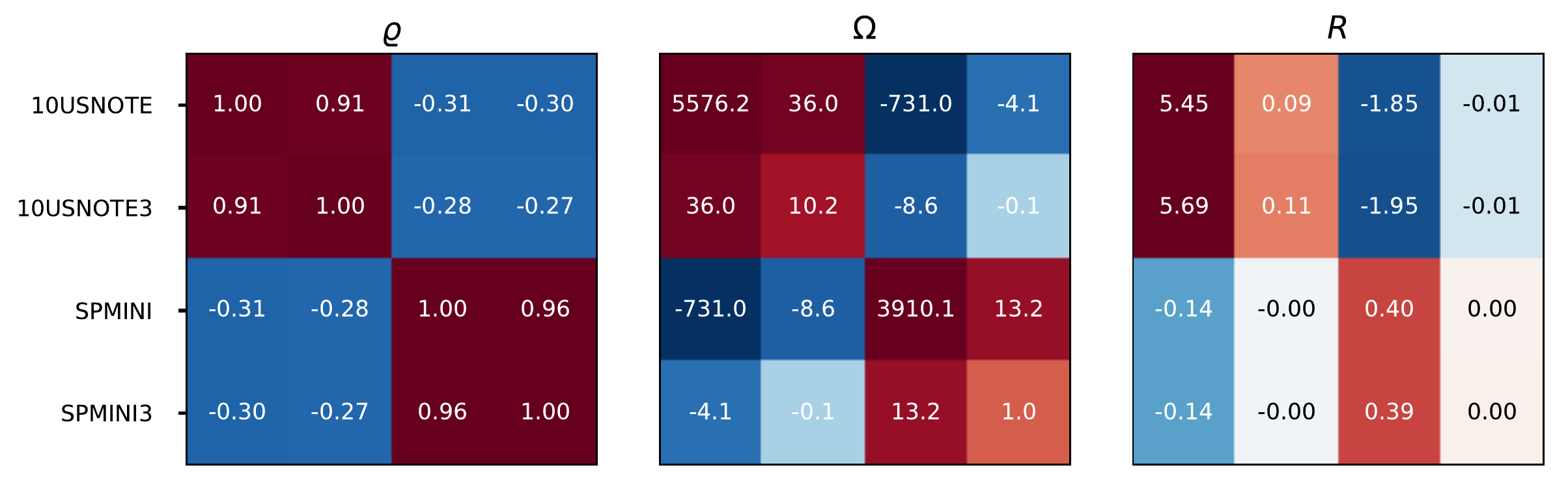}
    \caption{ \textbf{Estimates of  $\varrho$,  $\Omega$ and  $R$ for bonds and indices (in MUSD).}\\
    The return correlation matrix $\rho$ (left), order flow covariance matrix $\Omega$ (center) and response matrix $R$ (right) were estimated using 2016 data and computed on the 17th of August 2016. To highlight the amount of notional traded, order flow is reported in millions of exchanged dollars according to the average value of each contract on the 17th of August 2016. Basis points were accounted for, so that one traded unit of the futures contracts entitles the owner to one unit of the underlying.}
    \label{fig:10usnotes_price_corr}
\end{figure}

\paragraph*{Structure of $\rho$, $\Omega$ and $R$} \cref{fig:10usnotes_price_corr} shows the estimators of $\varrho$, $\Omega$ and $R$ for the 17th of August 2016. Contracts with the same underlying are strongly correlated. Thus, $\rho$ shows 2 by 2 blocks of strongly correlated contracts and an anti-correlation between bonds and futures. Liquidity is heterogeneous as front month contracts are more actively traded. In this configuration, the discriminating factor between models should be stability axioms rather than fragmentation axioms. 

\begin{figure}[H]
    \centering
    \includegraphics[width=0.8\linewidth]{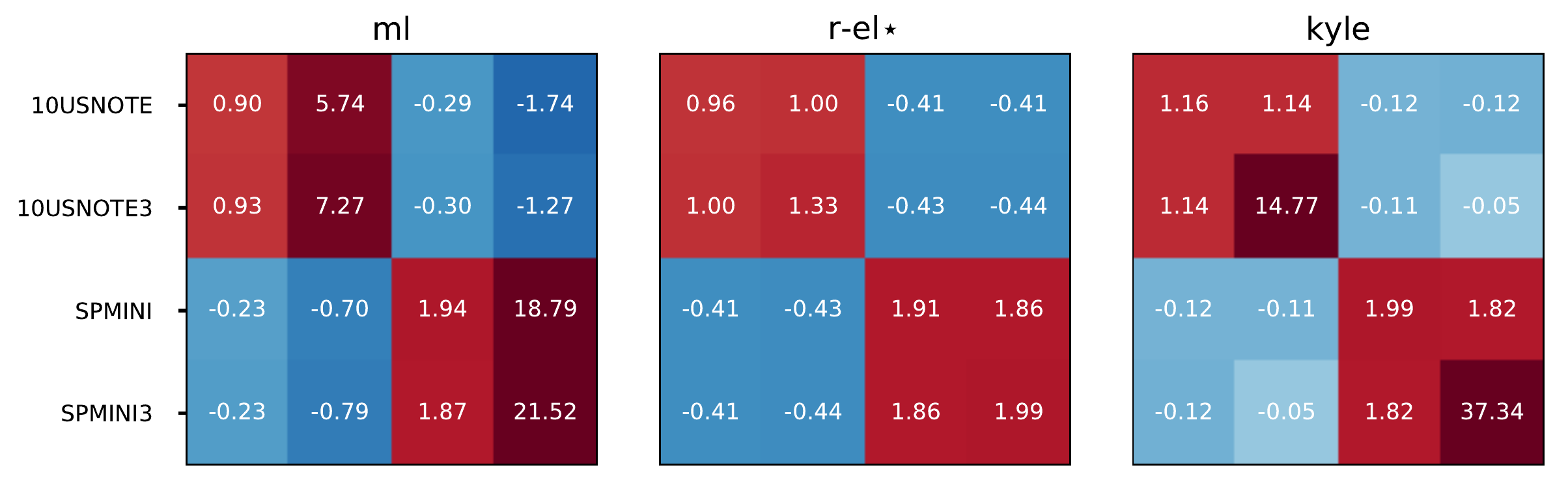}
    \caption{\textbf{Values of different cross-impact models for bonds and indices.}\\
    We report the values of the \rmle{} (left), \relm{}$\star$ (center) and \pkyle{} (right) cross-impact models for the covariances of the 17th of August 2016. Units are chosen to represent the relative price change in basis points ($10^{-4}$ of the asset price) by hundred  million USD worth of contract traded.}
    \label{fig:usnote_impact_risk}
\end{figure}

\paragraph*{Cross-impact models for bonds and indices} \cref{fig:usnote_impact_risk} shows the \rmle{}, \relm{}$\star$ and \pkyle{} models calibrated on bonds and indices. The \relm{} and \pkyle{} models are weakly cross-stable while the \rmle{} model is not. Thus \rmle{} assigns large impact to the less liquid contracts 10USNOTE3 and SPMINI3. Similarly, the self-stability of \relm{} explains the small impact predicted if one trades illiquid contracts. Reassuringly, all models correctly capture the negative index-bonds correlation.

\subsection{Stocks}
\label{app:data:stocks}

\paragraph*{Description of the dataset} We chose stocks which were in the S$\&$P500 index between January 2016 and December 2017. The resulting universe is made up of with 393 high market cap and liquid stocks. We chose such stocks to build a similar asset universe as in previous studies \cite{wang2015price,wang2016cross,wang2017grasping,Benzaquen2017DissectingAnalysis,pasquariello2015strategic}. We collect trades and quotes data between 2PM and 9:30PM UTC, removing the beginning and end of the trading period to focus on the intraday behavior of liquidity and volatility and circumvent intraday non-stationary issues. We collected trades and quotes data from January 2016 to December 2017, between  2PM and 9:30PM UTC, to focus on the intraday behavior of liquidity and volatility and circumvent intraday non-stationary issues. After filtering days for which data for one product was missing, we keep a total of 302 days (154 in 2016 and 148 in 2017). Some summary characteristics of our sample are presented in \cref{table:statistics_stocks}. The distribution of stocks in each sector is given in \cref{fig:sectors_barplot}.

\begin{table}[H]
    \centering
    \begin{tabular}{@{}lcccccccccc@{}}
    \multicolumn{1}{c}{} & \multicolumn{3}{c}{Quantile} \\
    \cmidrule{2-4}  & $10\%$ & $50\%$ & $90\%$ \\
    Relative tick size (in $\%$) & 1.6& 2.5& 4.6 \\ 
    Number of trades per day (in thousands) & 5.9& 12.6& 29.4 \\ 
    Daily turnover (in MUSD) & 28.5& 56.1& 116.2 \\ 
    \end{tabular}
    \caption{\textbf{Summary statistics for our sample of stocks.}}
    \label{table:statistics_stocks}
\end{table}

\begin{figure}[H]
    \centering
    \includegraphics[width=0.5\columnwidth]{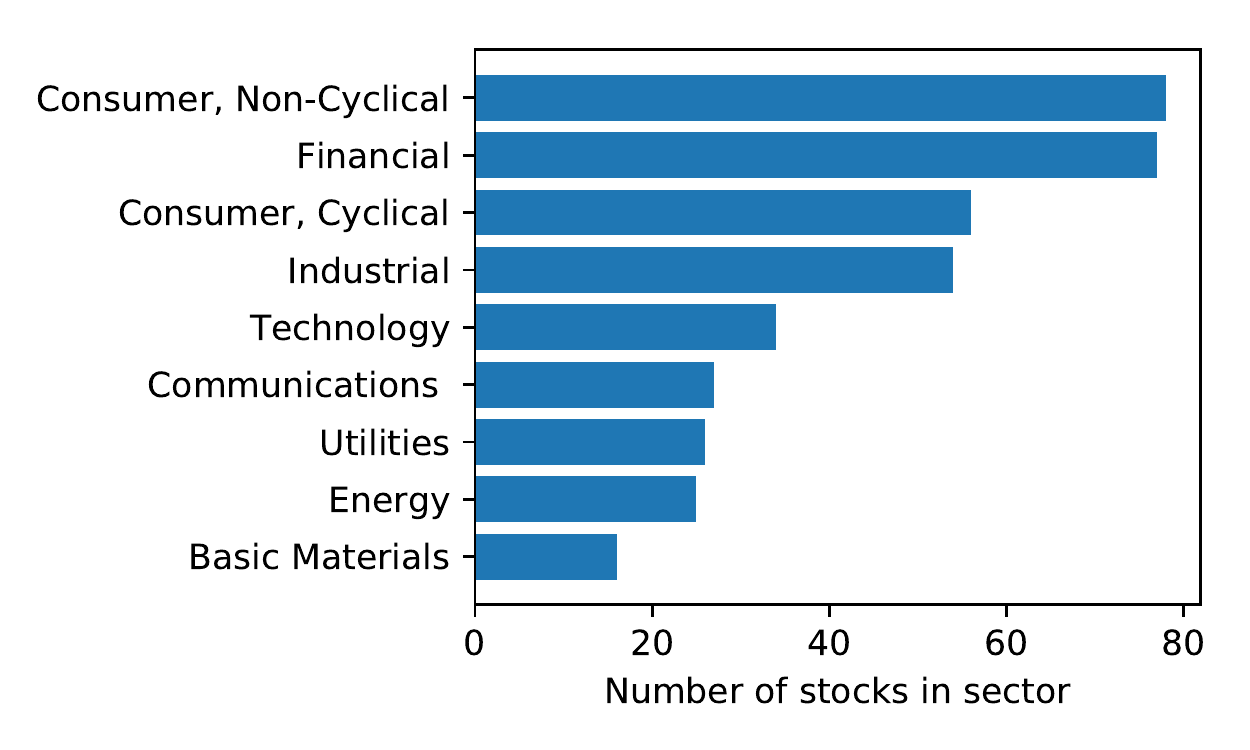}
    \caption{\textbf{Sector breakdown for the 393 of stocks used in the stocks dataset.}}
    \label{fig:sectors_barplot}
\end{figure}

\paragraph*{Pre-processing} To a lesser degree than on the previous datasets, the stock dataset shows non-stationarity in both volatility and liquidity. Thus we adopt the same estimation procedure for the local covariance estimators $\Sigma_t$ and $\Omega_t$ by assuming stationarity of the correlations $\varrho = \diag{\boldm{\sigma_t}}^{-1} \Sigma_t  \diag{\boldm{\sigma_t}}^{-1}$ and $\varrho_\Omega = \diag{\boldm{\omega_t}}^{-1} \Omega_t \diag{\boldm{\omega_t}}^{-1}$.

\begin{figure}[H]
    \centering
    \includegraphics[width=0.82\linewidth]{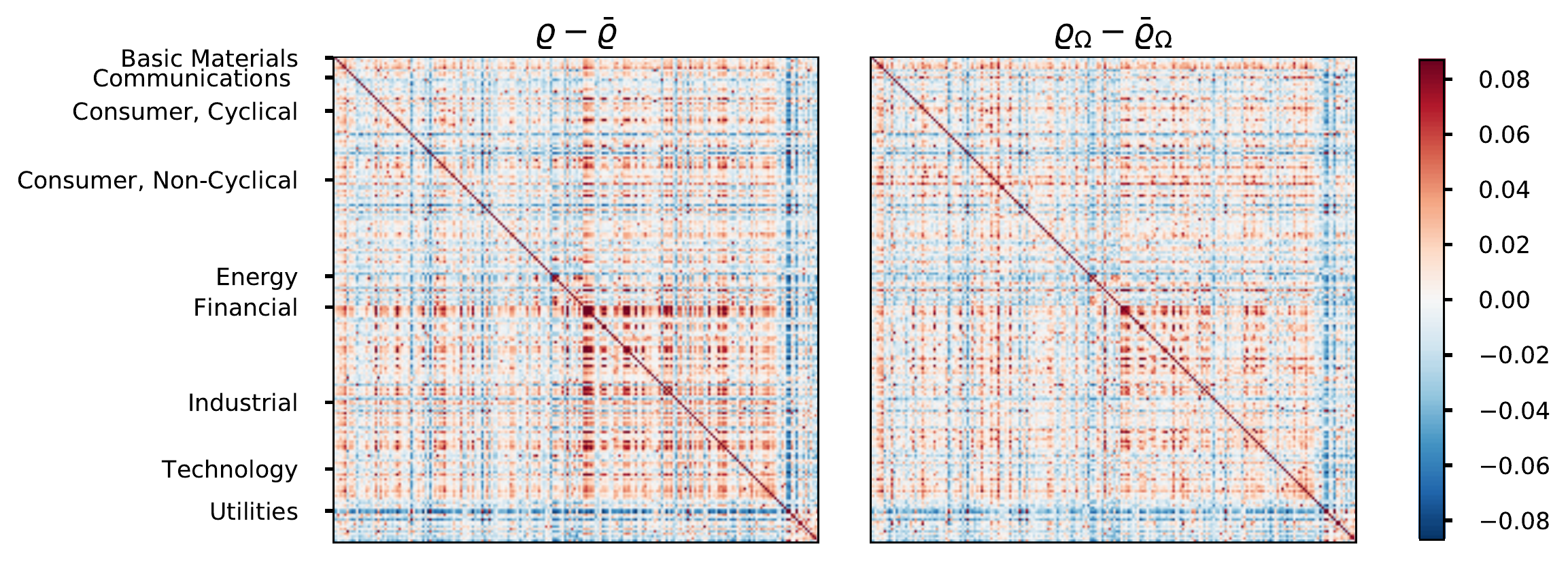}
    \caption{ \textbf{Estimated price and order flow correlation matrices $\varrho$, $\varrho_{\Omega}$ for stocks.}\\
    We represent the return correlation matrix $\rho$ (left), order flow correlation matrix $\varrho_\Omega$ (right) estimated on 2016. To highlight the amount of notional traded, order flow is reported in millions of exchanged dollars according to the average value of each contract on the 17th of August 2016. Correlation matrices were represented instead of covariance matrices due to the large volume heterogeneities between stocks. Stocks were grouped by sectors to highlight the blockwise structure of these matrices.}
    \label{fig:stocks_price_corr}
\end{figure}

\paragraph*{Structure of $\varrho$, $\varrho_\Omega$ and $R$} \cref{fig:stocks_price_corr} shows estimators of $\varrho$, $\varrho_\Omega$. We report correlations instead of covariances to highlight the blockwise structure of these matrices. For the same reasons, $R$ is not shown but presents a bandwise structure one expects from heterogeneity in liquidity. Pairwise price and order flow correlations between assets are small. Thus, improvement of cross-impact models over direct models should be lower than in previous applications. For more details about the structure of the price and volume covariance matrices, see~\cite{Benzaquen2017DissectingAnalysis}.

\begin{figure}[H]
    \centering
    \includegraphics[width=\linewidth]{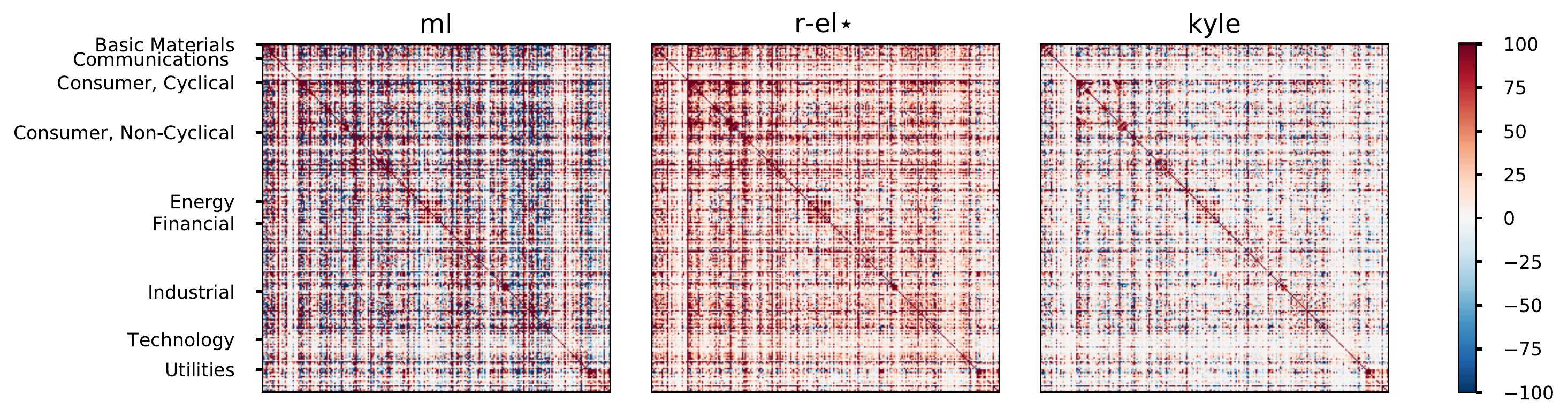}
    \caption{  \textbf{Values of different cross-impact models for stocks.} \\
     We report the values of the \rmle{} (left), \relm{}$\star$ (center) and \pkyle{} (right) cross-impact models. Units are chosen to represent the relative price change in basis points ($10^{-4}$ of the asset price) by hundred  million USD worth of instruments traded.}
     \label{fig:stocks_impact_risk}
\end{figure}

\paragraph*{Cross-impact models for stocks} \cref{fig:stocks_impact_risk} shows the \rmle{}, \relm{}$\star$ and \pkyle{} models calibrated on the stocks dataset. At first glance, each model appears to present a blockwise structure similar to that of $\varrho$, $\varrho_\Omega$. However, the \rmle{} model does not satisfy weak cross-stability and thus predicts large impact on liquid stocks if one trades illiquid stocks. By construction the \relm{}$\star$ model weighs most impact on the market mode. Finally, the \pkyle{} model looks like a symmetrized version of the \relm{}$\star$ model.

\section{Goodness-of-fit}
\label{app:goodness_of_fit}

We report in \cref{fig:gof} the numerical results of the goodness-of-fit tests run on each model and dataset.

\begin{landscape}

\renewcommand{\arraystretch}{1.1}

\begin{table}
\centering
\scalebox{0.75}{
\begin{tabular}{llccccccccccc}
    \toprule \multicolumn{1}{l}{Dataset} & \multicolumn{1}{l}{Score} & \multicolumn{1}{c}{\pdir} & \multicolumn{1}{c}{\pwhi} & \multicolumn{1}{c}{\pwhi$\star$} & \multicolumn{1}{c}{\pelm} & \multicolumn{1}{c}{\pelm$\star$} & \multicolumn{1}{c}{\pkyle} & \multicolumn{1}{c}{\rdir} & \multicolumn{1}{c}{\rmle} & \multicolumn{1}{c}{\relm} & \multicolumn{1}{c}{\relm$\star$} & \multicolumn{1}{c}{\rkyle} \\
     \cmidrule{1-13} \\
     Crude Futures \\
     & $\rsq_{\textnormal{in}}(I_{\sigma})$ &  $ 0.01 \pm 0.01 $&  $ 0.03 \pm 0.01 $&  $ 0.06 \pm 0.01 $&  $ 0.18 \pm 0.01 $&  $ 0.18 \pm 0.01 $&  $ 0.35 \pm 0.01 $&  $ 0.27 \pm 0.01 $&  $ 0.37 \pm 0.01 $&  $ 0.37 \pm 0.01 $&  $ 0.37 \pm 0.01 $&  $ 0.22 \pm 0.01 $
 \\
     & $\rsq_{\textnormal{out}}(I_{\sigma})$ &  $ 0.01 \pm 0.01 $&  $ 0.04 \pm 0.01 $&  $ 0.06 \pm 0.01 $&  $ 0.18 \pm 0.01 $&  $ 0.18 \pm 0.01 $&  $ 0.35 \pm 0.01 $&  $ 0.27 \pm 0.01 $&  $ 0.37 \pm 0.01 $&  $ 0.37 \pm 0.01 $&  $ 0.37 \pm 0.01 $&  $ 0.22 \pm 0.01 $
 \\
     & $\rsq_{\textnormal{in}}(J_{\sigma})$ &  $ 0.33 \pm 0.01 $&  $ 0.32 \pm 0.01 $&  $ 0.22 \pm 0.01 $&  $ 0.27 \pm 0.01 $&  $ 0.27 \pm 0.01 $&  $ 0.46 \pm 0.01 $&  $ 0.40 \pm 0.01 $&  $ 0.45 \pm 0.01 $&  $ 0.46 \pm 0.01 $&  $ 0.46 \pm 0.01 $&  $ 0.31 \pm 0.01 $
 \\
     & $\rsq_{\textnormal{out}}(J_{\sigma})$ &  $ 0.33 \pm 0.01 $&  $ 0.32 \pm 0.01 $&  $ 0.22 \pm 0.01 $&  $ 0.27 \pm 0.01 $&  $ 0.27 \pm 0.01 $&  $ 0.46 \pm 0.01 $&  $ 0.40 \pm 0.01 $&  $ 0.45 \pm 0.01 $&  $ 0.46 \pm 0.01 $&  $ 0.46 \pm 0.01 $&  $ 0.31 \pm 0.01 $
 \\
     & $\rsq_{\textnormal{in}}(\Sigma^{-1})$ &  $-\infty$ &  $ -0.05 \pm 0.02 $&  $ -0.01 \pm 0.02 $&  $ 0.07 \pm 0.02 $&  $ 0.07 \pm 0.02 $&  $ 0.29 \pm 0.02 $&  $-\infty$ &  $ 0.32 \pm 0.02 $&  $ 0.31 \pm 0.02 $&  $ 0.31 \pm 0.02 $&  $ 0.16 \pm 0.02 $
 \\
     & $\rsq_{\textnormal{out}}(\Sigma^{-1})$ &  $-\infty$ &  $ -0.05 \pm 0.02 $ &  $ -0.01 \pm 0.02 $ &  $ 0.07 \pm 0.02 $ &  $ 0.07 \pm 0.02 $ &  $ 0.29 \pm 0.02 $ &  $-\infty$ &  $ 0.31 \pm 0.02 $ &  $ 0.31 \pm 0.02 $ &  $ 0.31 \pm 0.02 $ &  $ 0.16 \pm 0.02 $ 
 \\
     \cmidrule{1-13} \\
     Bonds and indices \\
     & $\rsq_{\textnormal{in}}(I_{\sigma})$ &  $ -0.11 \pm 0.02 $&  $ 0.03 \pm 0.02 $&  $ 0.05 \pm 0.02 $&  $ 0.19 \pm 0.01 $&  $ 0.02 \pm 0.02 $&  $ 0.38 \pm 0.01 $&  $ 0.23 \pm 0.01 $&  $ 0.40 \pm 0.01 $&  $ 0.38 \pm 0.01 $&  $ 0.27 \pm 0.01 $&  $ 0.25 \pm 0.01 $ \\
     & $\rsq_{\textnormal{out}}(I_{\sigma})$ &  $ -0.11 \pm 0.02 $&  $ 0.03 \pm 0.02 $&  $ 0.04 \pm 0.02 $&  $ 0.19 \pm 0.01 $&  $ 0.02 \pm 0.02 $&  $ 0.38 \pm 0.01 $&  $ 0.23 \pm 0.01 $&  $ 0.40 \pm 0.01 $&  $ 0.38 \pm 0.01 $&  $ 0.27 \pm 0.01 $&  $ 0.24 \pm 0.01 $  \\
     & $\rsq_{\textnormal{in}}(J_{\sigma})$ &  $ 0.09 \pm 0.02 $&  $ -0.09 \pm 0.03 $&  $ -0.05 \pm 0.03 $&  $ 0.09 \pm 0.02 $&  $ -0.21 \pm 0.03 $&  $ 0.29 \pm 0.02 $&  $ 0.27 \pm 0.02 $&  $ 0.30 \pm 0.02 $&  $ 0.29 \pm 0.02 $&  $ 0.17 \pm 0.02 $&  $ 0.14 \pm 0.02 $ \\
     & $\rsq_{\textnormal{out}}(J_{\sigma})$ &  $ 0.09 \pm 0.02 $&  $ -0.10 \pm 0.03 $&  $ -0.05 \pm 0.03 $&  $ 0.09 \pm 0.02 $&  $ -0.21 \pm 0.03 $&  $ 0.29 \pm 0.02 $&  $ 0.27 \pm 0.02 $&  $ 0.30 \pm 0.02 $&  $ 0.29 \pm 0.02 $&  $ 0.17 \pm 0.02 $&  $ 0.14 \pm 0.02 $ \\
     & $\rsq_{\textnormal{in}}(\Sigma^{-1})$ &  $ -7.24 \pm 0.21 $&  $ -0.37 \pm 0.04 $&  $ -0.36 \pm 0.04 $&  $ -0.26 \pm 0.03 $&  $ -0.37 \pm 0.03 $&  $ 0.11 \pm 0.03 $&  $ -1.69 \pm 0.05 $&  $ 0.20 \pm 0.03 $&  $ 0.19 \pm 0.03 $&  $ 0.13 \pm 0.03 $&  $ 0.07 \pm 0.03 $ \\
     & $\rsq_{\textnormal{out}}(\Sigma^{-1})$ &  $ -7.23 \pm 0.21 $ &  $ -0.37 \pm 0.04 $ &  $ -0.36 \pm 0.04 $ &  $ -0.26 \pm 0.03 $ &  $ -0.37 \pm 0.03 $ &  $ 0.11 \pm 0.03 $ &  $ -1.71 \pm 0.05 $ &  $ 0.20 \pm 0.03 $ &  $ 0.19 \pm 0.03 $ &  $ 0.13 \pm 0.03 $ &  $ 0.07 \pm 0.03 $ \\
     \cmidrule{1-13} \\
     Stocks \\
     & $\rsq_{\textnormal{in}}(I_{\sigma})$  & $ 0.038 \pm 0.004 $ &  $ -0.025 \pm 0.004 $ &  $ 0.059 \pm 0.004 $ &  $ -0.631 \pm 0.010 $ &  $ -0.128 \pm 0.008 $ &  $ 0.343 \pm 0.003 $ &  $ 0.276 \pm 0.004 $ &  $ 0.373 \pm 0.003 $ &  $ 0.257 \pm 0.003 $ &  $ 0.236 \pm 0.004 $ &  $ 0.239 \pm 0.004 $\\
     & $\rsq_{\textnormal{out}}(I_{\sigma})$ &  $ 0.038 \pm 0.004 $ &  $ -0.031 \pm 0.004 $ &  $ 0.047 \pm 0.004 $ &  $ -0.642 \pm 0.010 $ &  $ -0.133 \pm 0.008 $ &  $ 0.336 \pm 0.003 $ &  $ 0.274 \pm 0.004 $ &  $ 0.358 \pm 0.003 $ &  $ 0.249 \pm 0.003 $ &  $ 0.227 \pm 0.004 $ &  $ 0.232 \pm 0.004 $  \\
     & $\rsq_{\textnormal{in}}(J_{\sigma})$ &  $ 0.732 \pm 0.006 $ &  $ -0.047 \pm 0.012 $ &  $ 0.277 \pm 0.010 $ &  $ -1.770 \pm 0.038 $ &  $ 0.727 \pm 0.005 $ &  $ 0.822 \pm 0.003 $ &  $ 0.480 \pm 0.010 $ &  $ 0.829 \pm 0.003 $ &  $ 0.661 \pm 0.005 $ &  $ 0.753 \pm 0.004 $ &  $ 0.788 \pm 0.004 $ \\
     & $\rsq_{\textnormal{out}}(J_{\sigma})$ &  $ 0.732 \pm 0.006 $&  $ -0.192 \pm 0.013 $&  $ 0.152 \pm 0.012 $&  $ -1.785 \pm 0.038 $&  $ 0.701 \pm 0.005 $&  $ 0.808 \pm 0.004 $&  $ 0.479 \pm 0.010 $&  $ 0.803 \pm 0.004 $&  $ 0.644 \pm 0.005 $&  $ 0.733 \pm 0.005 $&  $ 0.776 \pm 0.004 $ \\
     & $\rsq_{\textnormal{in}}(\Sigma^{-1})$ &  $ -0.311 \pm 0.004 $&  $ -0.061 \pm 0.003 $&  $ -0.056 \pm 0.003 $&  $ -0.262 \pm 0.005 $&  $ -0.369 \pm 0.008 $&  $ 0.214 \pm 0.003 $&  $ 0.180 \pm 0.003 $&  $ 0.215 \pm 0.003 $&  $ 0.126 \pm 0.004 $&  $ 0.090 \pm 0.004 $&  $ 0.082 \pm 0.004 $
 \\
     & $\rsq_{\textnormal{out}}(\Sigma^{-1})$ &  $ -0.293 \pm 0.004 $ &  $ -0.061 \pm 0.003 $ &  $ -0.056 \pm 0.004 $ &  $ -0.260 \pm 0.005 $ &  $ -0.360 \pm 0.008 $ &  $ 0.211 \pm 0.003 $ &  $ 0.180 \pm 0.003 $ &  $ 0.208 \pm 0.003 $ &  $ 0.124 \pm 0.004 $ &  $ 0.089 \pm 0.004 $ &  $ 0.081 \pm 0.004 $ \
\\
\end{tabular}}
\caption{\textbf{Goodness-of-fit scores for each model and dataset.\\} Goodness of fit was measured using two years of data sampled at a time interval of one minute. In-sample data was used to calibrate each cross impact model. Out-of-sample goodness of fit was obtained by applying the calibrated models on never seen before data. We reported as $\infty$ the scores of models which are numerically infinite, but due to clipping appear finite.}
\label{fig:gof}
\end{table}
\end{landscape}

\renewcommand{\arraystretch}{1}

\end{document}